\documentclass[12pt]{article}
\usepackage[margin = 1.5 in]{geometry}
\usepackage[utf8]{inputenc}
\usepackage{amsfonts}
\usepackage{amsmath}
\usepackage{amssymb}
\usepackage[english]{babel}
\usepackage{amsthm}
\usepackage{tikz}
\usepackage{float}
\usepackage[export]{adjustbox}
\usepackage{subcaption}
\usepackage{bbm}
\usepackage{comment}
\usepackage{natbib}
\usepackage{enumerate}
\usepackage{hyperref}
\usepackage{cleveref}
\usepackage{arydshln}
\setcitestyle{authoryear,open={(},close={)}}

\excludecomment{note}

\theoremstyle{definition}
\newtheorem{theorem}{Theorem}
\newtheorem{proposition}{Proposition}
\newtheorem{lemma}{Lemma}
\newtheorem{observation}{Observation}
\newtheorem{corollary}{Corollary}
\newtheorem*{definition}{Definition}

\theoremstyle{remark}
\newtheorem{remark}{Remark}
\newtheorem{example}{Example}

\DeclareMathOperator*{\argmax}{\arg\max}
\DeclareMathOperator*{\R}{ \mathbbm{R}}

\newcommand{\vst}{\vspace{3mm}}

\usepackage{inputenc}

\title{\vspace{-3em} Greedy Allocations and Equitable Matchings\thanks{I am grateful to Yeon-Koo Che, Federico Echenique, Chris Shannon, and numerous others for discussion and feedback. My thanks also to seminar participants at UNC, NYU, Rochester, UCLA, and U Chicago for helpful comments and conversations.}}
\author{Quitz\'{e} Valenzuela-Stookey\thanks{Department of Economics, UC Berkeley, \texttt{quitze@berkeley.edu}.}}
\date{October 7, 2022}

\begin{document}
\maketitle

\begin{center}
\vspace{-3em}
    \Large \textcolor{blue}{\href{https://drive.google.com/file/d/1IRfgHpEX9bqBCYEj8wD3YJt1S8RjnsR2/view?usp=sharing}{Click here for the latest version}}
\end{center}
\vst

\begin{abstract}
I provide a novel approach to characterizing the set of interim realizable allocations, in the spirit of \cite{matthews1984implementability} and \cite{border1991implementation}. The approach allows me to identify precisely why exact characterizations are difficult to obtain in some settings. The main results of the paper then show how to adapt the approach in order to obtain \textit{approximate} characterizations of the interim realizable set in such settings. 

As an application, I study multi-item allocation problems when agents have capacity constraints. I identify necessary conditions for interim realizability, and show that these conditions are sufficient for realizability when the interim allocation in question is scaled by $\frac{1}{2}$. I then characterize a subset of the realizable polytope which contains all such scaled allocations. This polytope is generated by a majorization relationship between the scaled interim allocations and allocations induced by a certain ``greedy algorithm''. I use these results to study mechanism design with equity concerns and model ambiguity. I also relate optimal mechanisms to the commonly used deferred acceptance and serial dictatorship matching algorithms. For example, I provide conditions on the principal's objective such that by carefully choosing school priorities and running deferred acceptance, the principal can guarantee at least half of the optimal (full information) payoff.   
\end{abstract}

\newpage

In an allocation problem, a designer specifies a rule for choosing among a set of alternatives as a function of agents' types. Specifically, consider a setting with a finite set $I$ of agents, each of whom has a type drawn from a finite set $T_i$. Let $T = T_1\times T_2 \times \dots \times T_I$ be the space of type profiles. An allocation is a function $q: T \rightarrow \mathbb{R}^N$ such that for each $t \in T$, $q(t) \in P^{XP}(t)$. The set-valued function $P^{XP}(\cdot)$ represents feasibility constraints on the assignment of alternatives in $\mathbb{R}^N$ to type profiles. The abstract allocation problem described above nests many problems in social choice, matching, and mechanism design. For example, in a single-item auction environment $q$ describes the probability that each agent receives the item, as a function of every agent's reported value for it.\footnote{This is only a part of the description of an auction; a full description includes the payments made by each agent. In general, however, payments are pinned down by the allocation rule \citep{myerson1981optimal}.} 

The above description makes no mention of incentives: in many settings agents privately observe their types, and may have incentives to misreport these to the designer in order to manipulate the assignment. Since the constraint set $P^{XP}$ is allowed to depend on $t$, the above formulation can capture settings in which ex-post (conditional on $t$) incentive constraints are imposed. Interim incentive constraints, however, cannot be characterized in this way.

Assume that there is a common prior $\mu$ on $T$, and let $\mu_i(\cdot|t_i)$ be the belief of agent $i$ over the types of the other agents, conditional on $i$ having type $t_i$. Interim incentive compatibility requires that agents be willing to report their types truthfully, given that they know only the allocation rule $q$, their own type, and the distribution $\mu_i(\cdot|t_i)$ over the set $T_{-i}$ of other agents' types. Ex-post incentive compatibility is sufficient to guarantee interim incentive compatibility, but is in general a more restrictive requirement.\footnote{There important settings in which ex-post and interim incentive compatibility are in fact equivalent, see \cite{manelli2010bayesian} and \cite{gershkov2013equivalence}. These results are closely related to the substance of the current paper.}  

Interim incentive compatibility for agent $i$ is a property of the \textit{interim allocation rule} $Q_i : T_i \rightarrow \mathbb{R}^N$ induced by the allocation rule $q$, which is defined by
\begin{equation}\label{eq:intro_interim}
    Q_i(\tau) := \mathbb{E}_{t_{-i}}\left[ q(t)| t_i = \tau \right] = \sum_{t_{-i} \in T_{-i}} q(t_{-i},\tau)\mu_i(t_{-i}|\tau).
\end{equation}
The interim allocation rule describes the distribution over outcomes, conditional on agent $i$'s type.   

It is convenient, for the purpose of designing a mechanism, to work directly in the space of interim allocation rules, rather than the space of allocation rules. There are two reasons for this. First, the interim allocation rule is a much simpler object than the allocation rule; the later is defined on the space of type profiles, a potentially large product space. Second, as noted above, interim incentive compatibility is fundamentally a property of the interim allocation rule. Aside from imposing ex-post incentive compatibility, which may be a stronger property than is desired, it is not obvious how to impose interim incentive compatibility directly on an allocation rule $q$.

In order to work directly in the space of interim allocation rules, one must know which allocation rules are in fact legitimate, or \textit{realizable}, in the sense that they are induced as the marginals of some allocation rule as in \cref{eq:intro_interim}. 

This paper makes two main contributions to the large literature dedicated to characterizing the set of realizable interim allocation rules in various settings, beginning with \cite{maskin1984optimal} and \cite{border1991implementation}.\footnote{In this literature, an allocation rule is sometimes referred to as an auction, and an interim allocation rule as a reduced-form auction.} First, I provide a new framework for characterizing interim realizability in a general setting, nesting most of those previously studied. The approach is based on two ways of representing a convex set (in this case, the set of realizable interim allocations) as both \textit{i.} the convex hull of it's extreme points (the $v$-representation), and \textit{ii.} the intersection of the set of half-spaces that contain it (the $h$-representation). Using this approach, I generalize the characterization results of \cite{che2013generalized}, by allowing for the constraints on the ex-post assignment to depend on the type profile (\Cref{thm:exact_submodular}). 

The use of the $h$ and $v$-representations to understand interim realizability is not new, these properties of convex sets have been used in various forms by \cite{border1991implementation}, \cite{border2007reduced}, \cite{gershkov2013equivalence}, and \cite{goeree2022geometric}, among others. While my approach characterizing interim realizability differs from these it is, in a sense, not a major technological innovation. It is valuable, however, for two reasons. First, I provide a ``modular'' approach to characterizing interim realizability. By separating out the key components of the characterization procedure, this approach makes it easy to see exactly when and how a parsimonious characterization of interim realizability is attainable. Second, the approach suggests a way forward when a parsimonious characterization is difficult to obtain.  

This ability to identify when and why a parsimonious characterization cannot be obtained leads to the second, and more significant, contribution of the paper: a method for identifying parsimonious \textit{approximate} characterizations of the set of realizable interim allocations, when a simple exact characterization is not available. To be precise, I adapt the procedure introduced to identify necessary and sufficient conditions for interim realizability, to instead identify conditions that are necessary and \textit{approximately} sufficient (\Cref{theorem:approximation_cover_char}). Approximate sufficiency means that there exists an $\alpha \in (0,1)$ such that if $Q$ satisfies the conditions, $\alpha Q$ is realizable. 

As an application, I study interim realizability in a one-to-one matching setting (many of the results generalize easily to many-to-many matching). \cite{gopalan2018public} show that no computationally tractable exact characterization of realizable interim allocations exists in this setting. Even abandoning computational considerations, theoretically meaningful characterizations of interim reliability with multiple items remain elusive (see \Cref{sec:multi_item} for a discussion of the literature). I identify the fundamental source of this difficulty, and then provide an approximate characterization. I then use this approximate characterization to answer applied matching questions. For example, I provide a partial answer to the question: how well does the commonly used deferred acceptance (DA) algorithm do at maximizing the designers objective? I show that by carefully choosing schools' priority rankings over students, DA guarantees at least half of the principal's full-information payoff (\Cref{thm:DA_guarantee}). I also study problems in which the principal seeks to elicit cardinal preferences.

The main technical result of this section is an approximate extension of Border's theorem to the multi-item setting (\Cref{thm:half_char}). The result also relates interim realizability to majorization by allocations induced by a certain ``greedy algorithm''. I then build on this to better understand the design of optimal allocation mechanisms under various principal objectives, including concerns for equity and robustness to model uncertainty, and relate optimal principal payoffs to that which can be achieved under the good properties approach (\Cref{sec:applications}). For example, I provide a partial answer to the question: how well does the deferred acceptance algorithm do at maximizing the principal's objective? I show that by carefully choosing schools' priority rankings over students, DA guarantees at least half of the principal's full-information payoff (\Cref{thm:DA_guarantee}). I also study problems in which the principal seeks to elicit cardinal preferences.

\section{Model}

I begin with a general description of the interim realizability problem, which nests existing models of single-item allocations with constraints \citep{che2013generalized}, public-goods problems as in \cite{goeree2022geometric}, and the multi-item setting which I focus on in \Cref{sec:multi_item}. 

Let $U = \{1,\dots,|U| \}$ be a finite set of \textit{units}, with typical element $u \in U$ (when it will not cause confusion, I also use $U$ to denote the number of units). Each unit $u$ has a type drawn from finite set $T_u$, with typical element $\tau$.\footnote{In a single item allocation problem, we would think of the units as the agents participating in the mechanism. In a multi-item problem, a unit will be a pair of ``agent'' + ``item'', as illustrated below. For the general formulation, I use the term ``unit'' to avoid confusion.} Let $T = T_1\times\dots\times T_N$ be the set of type profiles, with typical element $t$, where $t_u$ denotes unit $u$'s type in profile $t$. I refer to a realized type profile as a \textit{state}. Types are distributed according to the probability measure $\mu$ on $T$.\footnote{Independence of types will of course play a role when discussing incentives. With a single item, type independence can also be used to simplify the characterization of interim realizability, as in \cite{border2007reduced}. The extent to which this is possible more generally, for example with multiple items, remains an open question.} Let $\mu_u(t_{-u}|\tau)$ be the conditional distribution on $T_{-i}$ given $t_u = \tau$, and let $\mu_u^{\bullet}(\cdot)$ be the marginal distribution over $T_u$. I use the notation $t \sim (u,\tau)$ to denote that $t_u = \tau$.

It is convenient to define the \textit{non-null disjoint union} of types, $\mathcal{T}^* = \{(u,\tau) : \mu_u^{\bullet}(\tau) >0 \}$. That is, $\mathcal{T}^*$ is the set of unit-type pairs that might realize. Finally, for any $A \subseteq \mathcal{T}$ and $t\in T$, let $S(A,t) := \{ u : (u,t_u) \in A \}$.

An \textit{allocation rule}, or simply allocation, is a map $q:U\times T \rightarrow \mathbb{R}_+$ such that $q(\cdot,t) \in P^{XP}(t)$ for some polytope $P^{XP}(t)$. I refer to $P^{XP}(t)$ as the \textit{ex-post assignment polytope}, or just ex-post polytope, in state $t$. I refer to a vector $\rho \in P^{XP}(t)$ as an \textit{assignment} in state $t$. To simplify the exposition, I focus the discussion on the case where $P^{XP}(t) = P^{XP}$ for all $t$. \Cref{sec:P_I} makes it clear that the discussion extends easily to the case in which the ex-post polytope varies with $t$. The results apply to the general case where the ex-post polytope is type dependant. 

The following are examples of settings which can be modeled with this framework. 

\vst
\noindent\textit{Example 1: single-item with set constraints}. There is one unit of an infinitely divisible item, to be allocated among $I$ agents. In this setting each unit is an agent, i.e. $U = I$. The allocation $q(t,u)$ is the quantity of the item assigned to agent $u$ in state $t$. There may be constraints on the ex-post allocation: $P^{XP}(t)$ is the set of $\rho : U \rightarrow \mathbb{R}$ satisfying 
\begin{equation}\label{eq:constrained_allocation}
    L(A,t) \leq \sum_{i \in A} \rho(i) \leq C(A,t) \quad \forall \ A \subset U \text{ and } \forall t\in T,
\end{equation}
where $C(\cdot,t): 2^U \rightarrow \mathbb{R}_+$ and $L(\cdot,t): 2^U \rightarrow \mathbb{R}_+$ are such that $C(\varnothing,t) = L(\varnothing,t) = 0$ for all $t$. I refer to $C$ as the upper-constraint function, and $L$ as the lower-constraint function.
(Equivalently, there is a single indivisible item, and $q(t,u)$ is the probability that $u$ gets the item in state $t$. Under this interpretation, the constraints $L$ and $C$ are imposed not on ex-post allocations, but on the expected allocation conditional on the type profile). \cite{che2013generalized} study the special case in which $C$ and $L$ do not depend on $t$. The classic setting of \cite{border1991implementation}, in which the only constraints come from the unit supply of the item, corresponds to $L(A) = 0$ and $C(A) = 1$ for all $A \subset I$. 

\vst
\noindent\textit{Example 2: multiple items}. There are $N$ items and $I$ agents. In this case, a unit is a pair $(j,n)$ or an agent $j$ and an item $n$. To map this into the current framework, we impose the restriction that $t_{(j,n)} = t_{(j,n')}$ for all $j\in I, n,n' \in N$, and $t \in T$. The ex-post polytope for this setting is discussed in detail in \Cref{sec:mainresults}. 

\vst 
\noindent\textit{Example 3: public-goods problems \citep{goeree2022geometric}}. Consider a public goods problem with $I$ agents and $N$ alternatives. A unit is an alternative-agent pair. The ex-post polytope is given by the set of $\rho: I\times N \rightarrow \mathbb{R}_+$ such that 
\begin{enumerate}[i.]
    \item $\rho(i,n) = \rho(j,n) $ for all $n\in N$ and $i,j \in I$.
    \item $\sum_{(i,n) \in I\times N} \rho(i,n) \leq 1$. 
\end{enumerate}

\vst
Given an allocation rule, $q$, we can obtain an \textit{interim allocation} rule $Q: \mathcal{T}^* \rightarrow \R_+$ by averaging each unit's allocation over the types of other units: 
\begin{equation}\label{eq:general_interim_def}
    Q(u,\tau) := E_{t_{-u}\sim \mu_u(\cdot|\tau)}[q(u,t_{-u},\tau)] = \sum_{t_{-u} \in T_{-u}} q(u,t_{-u},\tau) \mu_u(t_{-u}|\tau).
\end{equation}
Conversely, given a function $Q:\mathcal{T}^* \rightarrow \R_+$, we say that $Q$ is \textit{realizable} if there exists some allocation $q$ that induces it, i.e. such that \cref{eq:general_interim_def} holds. Let $\mathcal{I}$ be the set of realizable interim allocations. The goal is to obtain a more convenient characterization of $\mathcal{I}$, in particular one that does not include an existential qualifier. Moreover, we would like this characterization to be simple, in the sense that it is relatively easy to check whether a given $Q$ is in $\mathcal{I}$. Ideally, the characterization should also facilitate optimization over $\mathcal{I}$. 

\subsection{Preliminary observations}

The set of realizable interim allocations, $\mathcal{I}$, is a polytope. This is immediate from the fact that $Q \in \mathcal{I}$ is a linear function of allocation rule $q$ that realizes it, and $q(\cdot,t)$ is constrained to a polytope for all $t \in T$. The strategy for obtaining simple characterizations of $\mathcal{I}$ makes use of three basic facts about polytopes. 

\vst
\noindent\textbf{Fact 1.} Every linear function $\lambda : \mathcal{I} \rightarrow \mathbbm{R}$ obtains its maximum on an extreme point of $\mathcal{I}$.

\vst
\noindent\textbf{Fact 2.} Every extreme point of $\mathcal{I}$ is the unique maximizer of some linear function on $\mathcal{I}$.

\vst
\noindent\textbf{Fact 3.} $Q \in \mathcal{I}$ if and only if for any linear function $f$ on $\mathcal{I}$ there is an extreme point $Q'$ of $\mathcal{I}$ such that $\lambda(Q) \leq \lambda(Q')$. (Separating hyperplane theorem). 

\vst
Identify each linear function on $\mathcal{I}$ with a function on $\mathcal{T}^*$: for a linear function $\lambda$ on $\mathcal{T}^*$, I abuse notation and write $\lambda(Q) = \sum_{(u,\tau) \in \mathcal{T}^*} \lambda(u,\tau)Q(u,\tau)$. Let $\Lambda$ be the space real functions on $\mathcal{T}^*$, normalized so that $\left|\frac{\lambda(u,\tau)}{\mu_u^{\bullet}(\tau)}\right| \leq 1$ for all $(u,\tau) \in \mathcal{T}^*$ (normalization is without loss of generality, and this one happens to be convenient). 

\vst
\begin{observation}\label{remark:trivial}
Facts 3 is just Fact 1 plus the separating hyperplane theorem. In other words, $\mathcal{I}$ is equal to the intersection of the halfspaces that contain it. This yields a trivial characterization of $\mathcal{I}$
\begin{equation}\tag{trivial characterization}
    Q \in \mathcal{I} \text{ iff } \lambda(Q) \leq \max_{Q' \in \mathcal{I}}\lambda(Q') \text{ for all } \ \lambda \in \Lambda. 
\end{equation}
The (convex) function $\lambda \mapsto \max_{Q' \in \mathcal{I}} \lambda(Q')$ is know as the \textit{support function}.
\end{observation}

The objective is to simplify the trivial characterization, or baring this, obtain a simple but approximate characterization of $\mathcal{I}$ by exploiting the convex structure of $\mathcal{I}$. 

\section{The simple geometry of interim realizability}\label{sec:geometry}

I begin by outlining the high level approach to to obtaining exact and approximate characterizations of $\mathcal{I}$. This makes clear the steps involved in going from the trivial characterization of $\mathcal{I}$ in \Cref{remark:trivial} to a more parsimonious characterization. This approach helps clarify why existing characterizations, such as \cite{border2007reduced} and \cite{che2013generalized}, take the form that they do (we will be able to give simple proofs of these results). More importantly, this framework makes it clear how existing results can be extended, why parsimonious characterizations remain elusive in some settings (such as the one-to-one matching problem), and how to go about finding approximate characterizations in these cases.

\subsection{Exact characterization}\label{sec:exact_outline}

An $h$-representation of a polytope $I$ consists of a set of half-spaces the intersection of which is exactly $I$.\footnote{The general insight I exploit in this section, that the polytope $\mathcal{I}$ can be understood via its support function, is not new. See for example \cite{vohra2011mechanism} and \cite{goeree2022geometric}, the latter of which is most similar to the current treatment. The approach to characterizing the $h$-representation (or support function) of $\mathcal{I}$ differs from \cite{goeree2022geometric} however. One way to understand this difference is that, by studying equivalence covers (see below), I characterize the sections of $\Lambda$ over which the support function is linear. The value of formulating the exact characterization in this way is that the approach extends naturally to approximate characterizations.} The trivial characterization in \Cref{remark:trivial} is an $h$-representation, however it is unsatisfying as a characterization for two reasons
\begin{enumerate}
    \item It requires checking infinitely many $\lambda$'s.
    \item For each $\lambda$, it requires maximizing over $\mathcal{I}$. 
\end{enumerate}
Ideally, we would like a characterization via an $h$-representation of the form
\begin{equation}\label{eq:border_type}
    Q \in \mathcal{I} \Longleftrightarrow \lambda(Q) \leq b (\lambda) \quad \forall \ \lambda \in \Lambda^*.
\end{equation}
for some ``small'' set $\Lambda^* \subset \Lambda$ and some known function $b:\Lambda^* \rightarrow \mathbb{R}$. In other words, we want a parsimonious $h$-representation. 

It turns out that the key to obtaining such a representation is to first identify the extreme points of $\mathcal{I}$. This yields the so-called \textit{v-representation} of $\mathcal{I}$ (a polytope is the convex hull of its extreme points). We then use the $v$-representation to achieve a parsimonious $h$-representation by
\begin{enumerate}
    \item Identifying a finite set $\Lambda^*$ of normal vectors for an $h$-representation of $\mathcal{I}$.
    \item Characterizing the function $b(\lambda) := \max_{Q' \in \mathcal{I}} \lambda(Q')$.
\end{enumerate}

As will become clear, the theorems of \cite{border1991implementation}, \cite{border2007reduced}, and \cite{che2013generalized}, among others, are precisely about obtaining such an $h$-representation of $\mathcal{I}$.

To begin, assume that we have characterized the extreme points of $\mathcal{I}$. An exact characterization is easy to obtain precisely when the extreme points of $\mathcal{I}$ admit a simple description. In \Cref{sec:P_I} I then discuss how the structure of interim allocations simplifies the problem of characterizing extreme points of $\mathcal{I}$: this problem reduces to that of characterizing extreme points of $P^{XP}$. 

We now use the characterization of $ext(\mathcal{I})$ to identify a subset of $\Lambda$ which identifies all supporting hyperplanes of $\mathcal{I}$. 

\begin{definition}
An \textit{equivalence cover} of $\Lambda$ is collection $\{E(Q^*)\}_{Q^* \in ext(\mathcal{I})}$ such that
\begin{enumerate}
    \item $\{E(Q^*)\}_{Q^* \in ext(\mathcal{I})}$ covers $\Lambda$, i.e. $\Lambda \subset \cup_{Q^* \in ext(\mathcal{I})} E(Q^*)$.
    \item For all $Q^* \in ext(\mathcal{I})$
    \begin{equation}\label{eq:equivalence_cone}
   E(Q^*)\subset \left\{ \lambda \in \Lambda : \max_{Q' \in \mathcal{I}} \lambda(Q') = \lambda(Q^*) \right\}
\end{equation}
\end{enumerate}
\end{definition}

Given an equivalence cover $\{E(Q^*)\}_{Q^* \in ext(\mathcal{I})}$, refer to $E^*(Q)$ as the \textit{equivalence set} of $Q^*$. In other words, $E(Q^*)$ is a subset of the normal vectors corresponding to hyperplanes that bind at $Q^*$.\footnote{One way to obtain an equivalence cover of $\Lambda$ is to let $E(Q^*) = \left\{ \lambda \in \Lambda : \max_{Q' \in \mathcal{I}} \lambda(Q') = \lambda(Q^*) \right\}$. In this case the collection $\{ E(Q^*)\}$ will cover $\Lambda$ by \Cref{remark:trivial}. However given that there may be significant overlap in the set of binding constraints across extreme points, it is convenient to allow for smaller covers.}

Return now to the trivial characterization of $\mathcal{I}$, which can be written as follows: $Q \in \mathcal{I}$ iff
\begin{equation*}
    \min_{\lambda \in \Lambda} \max_{Q' \in \mathcal{I}} \lambda(Q') - \lambda(Q) \geq 0.
\end{equation*}
Since $\lambda \mapsto \max_{Q' \in \mathcal{I}} \lambda(Q') - \lambda(Q)$ is the upper envelope of linear functions, and thus convex, the minimizing $\lambda$ will generally be interior. However once we know the set of extreme points $ext(\mathcal{I})$ and identify an equivalence cover, we just need to check that for each extreme point $Q^* \in ext(\mathcal{I})$ and each $\lambda \in E(Q^*)$, we have $\lambda(Q) \leq \lambda(Q^*)$. In other words, for each $Q^*$ we need to check
\begin{align}
    \min_{\lambda \in E(Q^*)} &\Big\{ \max_{Q' \in \mathcal{I}}\big\{\lambda(Q') - \lambda(Q)\big\}\Big\} \geq 0 \label{eq:affine_on_cones1}\\
    & \Longleftrightarrow  \min_{\lambda \in E(Q^*)} \Big\{ \lambda(Q^*) - \lambda(Q)\Big\} \geq 0 \label{eq:affine_on_cones2}
\end{align}The benefit of having identified the equivalence set associated with each extreme point is illustrated by the equivalence in \Cref{eq:affine_on_cones2}. The objective $\lambda \mapsto \max_{Q' \in \mathcal{I}}\big\{\lambda(Q') - \lambda(Q)\big\}$ is convex on $\Lambda$, but it is affine on each equivalence cone $E(Q^*)$. As a result, for each equivalence cone, $E(Q^*)$, we only need to check $\lambda(Q) \leq \lambda(Q^*)$ for $\lambda \in ext(E(Q^*))$, the extreme points of $E(Q^*)$. 

Let $\Lambda^* := \cup_{Q^*\in ext(\mathcal{I})} ext(E(Q^*))$ be the union of the extreme points of the equivalence sets. For any $\lambda \in \Lambda^*$, let $b^*(\lambda) = \{\lambda(Q^*) : \lambda \in E(Q^*)\}$. Then by construction, $b^*$ is a real-valued function. 

\begin{lemma}\label{lem:equivalencecover_char}
Given an equivalence cover $\{E(Q^*)\}_{Q^* \in ext(\mathcal{I})}$, define $\Lambda^*$ and $b^*$ as above. Then $Q \in \mathcal{I}$ if and only if
\begin{equation*}
    \lambda(Q) \leq b^*(\lambda) \quad \forall \ \lambda \in \Lambda^*.
\end{equation*}
\end{lemma}

\begin{remark}
\Cref{lem:equivalencecover_char} can be understood as a way to characterize the support function of $\mathcal{I}$, which is defined as the function $\lambda \rightarrow \max \{ \lambda(Q) : Q \in \mathcal{I}\}$. The ``modular'' approach to characterizing the support function, whereby we first identify the extreme points of $\mathcal{I}$ and an equivalence cover, is a convenient way to decompose the problem by identifying subsets of $\Lambda$ over which the support function is linear. Moreover, this approach can be adapted to obtain a parsimonious approximate characterization in case where the support function is complicated and/or not easily characterized. This is shown in \Cref{sec:approx_outline}. 
\end{remark}

The approach to characterizing $\mathcal{I}$ by identifying an equivalence cover is illustrated in \Cref{fig:exact}. The set $\mathcal{I}$ is the grey shaded region (the shape of this polytope is not important, the figure is just meant to illustrate the general procedure). For an extreme point of $\mathcal{I}$, say $Q_1$, there are three normal vectors illustrated for which $Q_1$ is maximal in $\mathcal{I}$: $\lambda_1,\lambda_2$ and $\lambda_4$. For any $j \in \{1,2,4\}$, and any $Q$, we have $\lambda_j(Q) \leq \max_{Q' \in \mathcal{I}}\lambda_{j}(Q') \Longleftrightarrow \lambda_j(Q) \leq \lambda_{j}(Q_1)$. The same holds for any $\lambda \in co(\{\lambda_1,\lambda_2,\lambda_4)$. Since in $co(\{\lambda_1,\lambda_2)$, it suffices to check only $\lambda_1$ and $\lambda_2$.

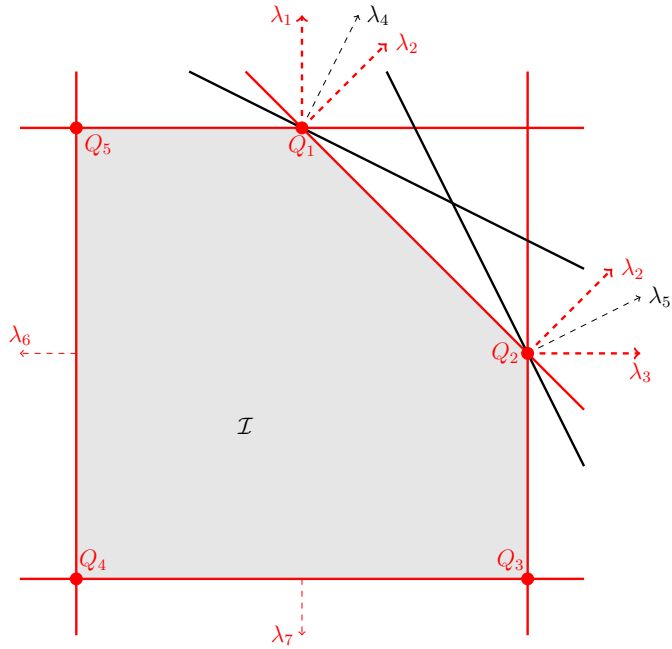
\begin{figure}
    \centering
    \scalebox{.75}{
    \begin{tikzpicture}
    
    \draw[fill = gray!20] (1,1) -- (1,9) -- (5,9) -- (9,5) -- (9,1) -- (1,1);
    
    \draw[very thick, red] (1,0) -- (1,10);
        \draw[dashed, red] [->] (1,5) -- (0,5) node[anchor = south] {$\lambda_6$};
    \draw[very thick, red] (0,1) -- (10,1);
        \draw[dashed, red] [->] (5,1) -- (5,0) node[anchor = east] {$\lambda_7$};
    \draw[very thick, red] (0,9) -- (10,9);
        \draw[very thick, red, dashed] [->] (5,9) -- (5,11) node[anchor = east] {$\lambda_1$};
    \draw[very thick, red] (9,0) -- (9,10);
        \draw[very thick, red, dashed] [->] (9,5) -- (11,5) node[anchor = north] {$\lambda_3$};
    \draw[very thick, red] (4,10) -- (10,4);
        \draw[very thick, red, dashed] [->] (5,9) -- (6.5,10.5) node[anchor = west] {$\lambda_2$};
        \draw[very thick, red, dashed] [->] (9,5) -- (10.5, 6.5) node[anchor = west] {$\lambda_2$};
    \draw[very thick] (3,10) -- (10,6.5);
        \draw[dashed] [->] (5,9) -- (6,11) node[anchor = west] {$\lambda_4$};
    \draw[very thick] (10,3) -- (6.5,10);
        \draw[dashed] [->] (9,5) -- (11,6) node[anchor = west] {$\lambda_5$};
        
    \draw[fill = gray!20, opacity = .1] (1,1) -- (1,9) -- (5,9) -- (9,5) -- (9,1) -- (1,1);
    \node [below] at (4,4) {$\mathcal{I}$};
    
    \filldraw[red] (5,9) circle (3pt) node[anchor = north] {$Q_1$};
    \filldraw[red] (9,5) circle (3pt) node[anchor = east] {$Q_2$};
    \filldraw[red] (9,1) circle (3pt);
        \node[red, above] at (8.7,1)  {$Q_3$};
    \filldraw[red] (1,1) circle (3pt);
        \node[red, above] at (1.3,1)  {$Q_4$};    
    \filldraw[red] (1,9) circle (3pt);
        \node[red, right] at (1,8.7)  {$Q_5$};        
    \end{tikzpicture}
    }
    
    \caption{Geometry of exact characterization}
    \label{fig:exact}
\end{figure}

I observe below (\Cref{lem:P_XP}) that the problem of characterizing the extreme points of $\mathcal{I}$ boils down to characterizing the extreme points of $P^{XP}$. For each $x^* \in ext(P^{XP})$, let $e(x^*) := \{ y \in [0,1]^U : y(x^*) = \argmax_{x \in P^{XP}} u(x) \}$.  Thus the simplicity of the characterization of $\mathcal{I}$ depends on 

\begin{enumerate}
    \item[S1.] How many extreme points $P^{XP}$ has (in other words, how large are the equivalence classes) and how easy are they to describe.
    \item[S2.] How many extreme points does each $e(x^*)$ have, and how easy are they to describe.\footnote{Equivalence classes for different extreme points may intersect, and so what really matters is the number of elements in the union of the extreme points of the equivalence classes.} 
\end{enumerate}

A parsimonious characterization of interim realizability holds only when conditions S1 and S2 are satisfied. When these conditions do not hold, it we face a trade-off between tractable and approximate characterizations. Before illustrating the technique for approximate characterization, I illustrate how the approach outlined here can be used to derive exact characterizations.  

\subsection{Approximate characterization}\label{sec:approx_outline}

A simple characterization of interim realizability only holds when conditions S1 and S2 are satisfied. If these conditions are do not hold, then it may be useful to look for a simple but approximate characterization. The key idea is to replicate as closely as possible the strategy detailed in \Cref{sec:exact_outline}, with modifications to account for the fact that extreme points of $\mathcal{I}$ do not admit a simple characterization. To do this, we relax the definition of equivalence cover.

\begin{definition}
A set of pairs $\{(Q^j, \hat{E}(Q^j))\}$, where each $Q^j \in \mathcal{I}$ and $\hat{E}(Q^j) \subset \Lambda$ is a polytope, is called an $\alpha$\textit{-approximation cover} if 
\begin{enumerate}
    \item $\{\hat{E}(Q^j)\}$ covers $\Lambda$
    \item For all $Q^j$,
    \begin{equation}\label{eq:approximation_cones}
        \lambda(Q^j) \geq \alpha \max_{Q'\in \mathcal{I}} \lambda(Q') \quad \forall \ \lambda \in \hat{E}(Q^j)
    \end{equation}
\end{enumerate}
\end{definition}
Given an $\alpha$-approximation cover $\{(Q^j, \hat{E}(Q^j)\}$, define the $\alpha$-\textit{approximation polytope} $\mathcal{I}^{\alpha} := co(\{Q^j\})$. Clearly $\mathcal{I}^{\alpha} \subset \mathcal{I}$.  

An equivalence cover $\{ E(Q^*) \}_{Q^ \in ext(\mathcal{I})}$ is a $1$-approximation cover, where we let $\{Q^j\} = ext(\mathcal{I})$ and $\hat{E}(Q^*j = E(Q^j)$. However for an equivalence cover, we know that $Q \in \mathcal{I}$ iff for all $Q^* \in ext(\mathcal{I})$
\begin{equation*}
    \min_{\lambda \in ext(E(Q^*))} \max_{Q' \in \mathcal{I}} \left\{ \lambda(Q') - \lambda(Q) \geq 0\right\}.
\end{equation*}
However given an $\alpha$-approximation cover with $\alpha < 1$, it is \textit{not} the case that
\begin{equation}\label{eq:approximate_characterization}
    \min_{\lambda \in ext(\hat{E}(Q^j))} \max_{Q' \in \mathcal{I}} \left\{ \lambda(Q') - \lambda(Q)\right\} \geq 0
\end{equation}
for all $Q^j$ implies that $Q \in \mathcal{I}$. The intersection of the half-spaces with normal vectors in $ext(\hat{E}(Q^j))$ we may be strictly larger than $\mathcal{I}$. However, the $\alpha$-approximation condition in \cref{eq:approximation_cones} guarantees that the condition in \cref{eq:approximate_characterization} is not ``too far'' from characterizing $\mathcal{I}$.

\begin{theorem}\label{theorem:approximation_cover_char}
Let $\{(Q^j, \hat{E}(Q^j)\}$ be an $\alpha$-approximation cover. If $Q \in \mathcal{I}$ then
\begin{equation}\label{eq:approximation_cover_char}
    \min_{\lambda \in ext(\hat{E}(Q^j))} \max_{Q' \in \mathcal{I}} \left\{\lambda(Q') - \lambda(Q) \right\} \geq 0 \quad \forall \ Q^j.
\end{equation}
Conversely, if \cref{eq:approximation_cover_char} holds then $\alpha Q \in \mathcal{I}^{\alpha} \subseteq \mathcal{I}$.
\end{theorem}
\begin{proof}
Necessity is obvious, since \cref{eq:approximation_cover_char} is implied by the trivial characterization of \Cref{remark:trivial}. We need to show that if $Q$ satisfies \cref{eq:approximation_cover_char} then for any $\lambda \in \Lambda$ there exists a vertex $Q^j$ of the $\alpha$-approximation polytope such that $\lambda(\alpha Q) \leq \lambda(Q^j)$. 

Let $\lambda'$ be arbitrary, and let $Q^j$ be such that $\lambda' \in \hat{E}(Q^j)$ (which exists since $\{\hat{E}(Q^j)\}$ covers $\Lambda$). Then we know $\lambda(Q^j) \geq \alpha \max_{Q'\in \mathcal{I}}\lambda(Q') \geq \alpha \lambda(Q)$ for all $\lambda \in ext(\hat{E}(Q^j))$, where the first inequality follows from definition of the $\alpha$-approximation cover, and the second by \cref{eq:approximation_cover_char}. Since $\lambda' \in \hat{E}(Q^j)$, we then have $\lambda'(Q^j) \geq \alpha \lambda'(Q) = \lambda'(\alpha Q)$, as desired. 
\end{proof}

The remaining questions are \textit{a}) whether one can identify an $\alpha$-approximation cover (ideally for $\alpha$ close to 1), and \textit{b)} if the function $\lambda \mapsto \max_{Q' \in \mathcal{I}} \lambda(Q')$ is easily characterized for $\lambda \in ext(\hat{E}(Q^j))$. In the next section I discuss how the structure of interim allocations simplifies this problem. However even without making use of this structure we can simplify the task of identifying an $\alpha$-approximation cover. To do this, re-write the condition in \cref{eq:approximation_cones} as
\begin{equation*}
    \min_{\lambda \in \hat{E}(Q^j)} \left\{ \lambda(Q^j) - \alpha \max_{Q' \in \mathcal{I}}\lambda (Q') \right\} \geq 0
\end{equation*}

Since $\lambda \rightarrow \left\{ \lambda(Q^j) - \alpha \max_{Q' \in \mathcal{I}}\lambda (Q') \right\}$ is concave, it suffices to check only the extreme points of $\hat{E}(Q^j)$.

\begin{lemma}\label{lem:alt_approximationcover}
Set of pairs $\{(Q^j, \hat{E}(Q^j))\}$ is an $\alpha$\textit{-approximation cover} if and only if $\{\hat{E}(Q^j)\}$ covers $\Lambda$, and for all $Q^j$
\begin{equation*}
    \lambda(Q^j) \geq \alpha \max_{Q' \in \mathcal{I}} \lambda(Q') \quad \ \lambda \in ext(\hat{E}(Q^j))
\end{equation*}
\end{lemma}

The approach outlined here to approximately characterizing $\mathcal{I}$ is illustrated in \Cref{fig:approx}. Here $Q_1$ such that $\lambda_j(Q_1) \geq \alpha \max_{Q'\in \mathcal{I}} \lambda_j(Q')$ for $j \in \{1,2,3,4,5\}$ and some $\alpha < 1$. Then if $\lambda_j(Q) \leq \max_{Q'\in \mathcal{I}} \lambda_j(Q')$ for $j \in \{1,2,3,4,5\}$ then $\lambda_j(\alpha Q) \leq \lambda_j(Q_1)$ for $j \in \{1,2,3,4,5\}$. In fact, the same conclusion holds if $\lambda_j(Q) \leq \max_{Q'\in \mathcal{I}} \lambda_j(Q')$ for $j \in \{1,4\}$, since these are the extreme points of $\{\lambda_1,\dots,\lambda_5\}$. Note that this set of normal vectors defines a larger polytope than $\mathcal{I}$ (the large blue square as opposed to the grey area). However we know that this polytope scaled by $\alpha$ is a subset of $co(\{Q_1,\dots,Q_5\})$, and thus a subset of $\mathcal{I}$. 

\begin{figure}
    \centering
    
    \scalebox{.75}{
    \begin{tikzpicture}
    
    \draw[thin, dashed, fill = gray!20] (1,1) -- (1,9) -- (5,9) -- (9,5) -- (9,1) -- (1,1);
    
    \draw[very thick, blue] (1,0) -- (1,1);
    \draw[very thick, dashed, blue] (1,1) -- (1,9);
    \draw[very thick, blue] (1,9) -- (1,10);
        \draw[dashed, blue] [->] (1,5) -- (0,5) node[anchor = south] {$\lambda_6$};
    \draw[very thick, blue] (0,1) -- (1,1);
    \draw[very thick, dashed, blue] (1,1) -- (9,1);
    \draw[very thick, blue] (9,1) -- (10,1);
        \draw[dashed, blue] [->] (5,1) -- (5,0) node[anchor = east] {$\lambda_7$};
    \draw[very thick, blue] (0,9) -- (10,9);
        \draw[very thick, blue, dashed] [->] (5,9) -- (5,11) node[anchor = east] {$\lambda_1$};
    \draw[very thick, blue] (9,0) -- (9,10);
        \draw[very thick, blue, dashed] [->] (9,5) -- (11,5) node[anchor = north] {$\lambda_3$};
    \draw[very thick, red] (4,10) -- (10,4);
        \draw[very thick, red, dashed] [->] (5,9) -- (6.5,10.5) node[anchor = west] {$\lambda_2$};
        \draw[very thick, red, dashed] [->] (9,5) -- (10.5, 6.5) node[anchor = west] {$\lambda_2$};
    \draw[very thick] (3,10) -- (10,6.5);
        \draw[dashed] [->] (5,9) -- (6,11) node[anchor = west] {$\lambda_4$};
    \draw[very thick] (10,3) -- (6.5,10);
        \draw[dashed] [->] (9,5) -- (11,6) node[anchor = west] {$\lambda_5$};
        
    \draw[fill = gray!20, opacity = .1] (1,1) -- (1,9) -- (5,9) -- (9,5) -- (9,1) -- (1,1);
    
    \filldraw[blue] (6.5,6.5) circle (3pt);
        \node[blue, below] at (6.4,6.4) {$Q_1$};
    \filldraw[blue] (9,1) circle (3pt);
        \node[blue, above] at (8.6,1)  {$Q_3$};
    \filldraw[blue] (1,1) circle (3pt);
        \node[blue, above] at (1.3,1)  {$Q_4$};    
    \filldraw[blue] (1,9) circle (3pt);
        \node[blue, right] at (1,8.6)  {$Q_5$}; 
        
    \draw[very thick, dashed, blue] (1,9) -- (6.5,6.5) -- (9,1);
    
    \draw[very thick, blue] (1.05,1.05) -- (1.05,5.5) -- (5.5,5.5) -- (5.5,1.05) -- (1.05,1.05);
    \end{tikzpicture}
    }
    
    \caption{Geometry of approximate characterization}
    \label{fig:approx}
\end{figure}
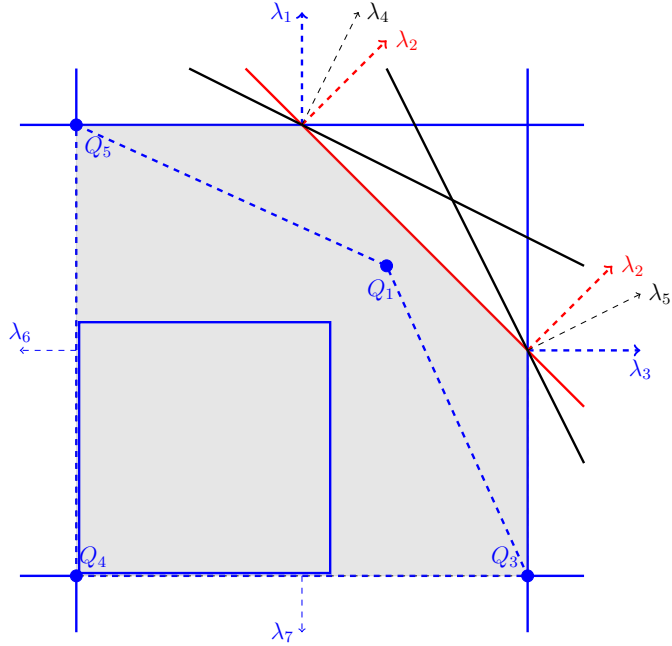

\section{Exploiting the structure of \texorpdfstring{$\mathcal{I}$}{}: state separability}\label{sec:P_I}

In this section I make a simple observation that the structure of interim allocations simplifies the problem of identifying equivalence covers and $\alpha$-approximation covers. The key observation is that maximizing linear functions on $\mathcal{I}$ reduces to maximizing linear functions on $P^{XP}(t)$. I refer to this property as \textit{state separability}.\footnote{This property is also observed by \cite{goeree2022geometric}.}

To see this, consider the program of maximizing a a linear function $\lambda$ on $\mathcal{I}$, given by
    \begin{align}\label{eq:P_I}
 \max_{Q \in \mathcal{I}} \lambda(Q) =  \max_{q} \sum_{(u,\tau) \in \mathcal{T}^*} \lambda(u,\tau)& \underbrace{\sum_{t_u \in T_{-u}} q(u,t_{-u},\tau)\mu_u(t_{-u}|\tau)}_{Q(u,\tau)} \\
        &s.t. \quad q(\cdot,t) \in P^{XP}(t) \quad \forall \ t \in T. \notag
    \end{align}
    We can rewrite the objective as
    \begin{equation}\label{eq:P_XP}
                \sum_{t \in T} \sum_{i \in I} \lambda(u,t_u) q(u,t) \mu_u(t_{-u}|\tau) 
        = \sum_{t \in T}  \mu(t) \sum_{i \in I} \frac{\lambda(u,t_u)}{\mu_u^{\bullet}(t_u)} q(u,t) 
    \end{equation}
Since the objective in \cref{eq:P_XP} separates across $t$, and the constraint $q(\cdot,t) \in P^{XP}(t)$, is also defined separately for each $t$, we just need to solve pointwise for each $t$. 

\begin{lemma}\label{lem:state_sep_exact}
$Q \in \argmax_{Q' \in \mathcal{I}} \lambda(Q')$ iff $Q$ is realized by a $q$ such that 
\begin{equation*}
    q(\cdot,t) \in \argmax_{\rho \in P^{XP}(t)} \sum_{u \in U} \frac{\lambda(u,t_u)}{\mu_u^{\bullet}(t_u)} \rho(i) 
\end{equation*}
for all $t\in T$. 
\end{lemma}

\subsection{Extreme points of \texorpdfstring{$\mathcal{I}$}{}}\label{sec:ex_post_extreme}

From \Cref{eq:P_XP} we can see that the problem of characterizing the extreme points of $\mathcal{I}$ reduces to the much simpler one of characterizing the extreme points of each es-post polytope $P^{XP}(t)$. Formally, for $\gamma \in \Gamma := [-1,1]^{U}$ define 
\begin{equation*}
    x^*(\gamma|t) := \argmax_{\rho \in P^{XP}(t)} \sum_{u \in U} \gamma(u) \rho(u).
\end{equation*}
Given $\lambda \in \Lambda$ and $t\in T$ let $\lambda(\cdot|t):= \left(\frac{\lambda(1,t_1)}{\mu_1^{\bullet}(t_1)},\frac{\lambda(2,t_2)}{\mu_2^{\bullet}(t_2)} ,\dots ,\frac{\lambda(|U|,t_{|U|})}{\mu_{|U|}^{\bullet}(t_{|U|})}\right)$ be the weights induced by $\lambda$ in state $t$. Given the maintained normalization of $\Lambda$, $\lambda(\cdot|t) \in \Gamma$. 

\begin{lemma}\label{lem:P_XP}\hphantom{.}
\begin{enumerate}
    \item $q$ solves the program in (\ref{eq:P_I}) if and only if $q(\cdot,t) \in x^*(\lambda(\cdot|t))$ for all $t$.
    \item If $Q$ is an extreme point of $\mathcal{I}$ then $Q$ is realized by a $q$ such that $q(\cdot,t)$ is an extreme point of $P^{XP}(t)$ for all $t$.\footnote{The converse need not be true: for $t \neq t'$ the vectors $(\lambda(u,t_u))_{u \in U}$ and $(\lambda(u,t'_i))_{u\in U}$ cannot be set independently when there is some $i$ such that $t_u = t'_i$. In any case, \Cref{lem:P_XP} is sufficient for our purposes.}
\end{enumerate}
\end{lemma}

In other words, once we understand the function $x^*$, the problem of characterizing the extreme points of $\mathcal{I}$ becomes trivial. 

\subsection{Approximation covers}

The state separability property observed in \Cref{lem:state_sep_exact} also simplifies the task of identifying $\alpha$-approximation covers: it is sufficient to identify $\alpha$-approximation covers in the ex-post polytope. 

I abuse notation and for $\gamma\in \Gamma$ write $\gamma(\rho) = \sum_{u\in U}\gamma(u)\rho(u)$. Say that $\{(\rho^j,\hat{e}(\rho^j)\}$ is an \textit{ex-post} $\alpha$-\textit{approximation cover} if

\begin{definition}
A set of pairs $\{(\rho^j,\hat{e}(\rho^j)\}$, where each $\rho^j \in P^{XP}$ and $\hat{e}(\rho^j) \subset \gamma$ is a polytope, is called an \textit{ex-post} $\alpha$\textit{-approximation cover} if 
\begin{enumerate}
    \item $\{\hat{e}(\rho^j)\}$ covers $\Gamma$, and
    \item for all $q^j$,
    \begin{equation}\label{eq:expost_approximation_cones}
        \gamma\rho^j) \geq \alpha \max_{\rho' \in P^{XP}} \gamma(\rho') \quad \forall \ \gamma \in \hat{e}(\rho^j)
    \end{equation}
\end{enumerate}
\end{definition}

\begin{lemma}\label{lem:expost_approximationcover}
Given an ex-post $\alpha$-approximation cover $\{(\rho^j,\hat{e}(\rho^j)\}$, we can create an (ex-ante) $\alpha$-approximation cover $\{ (Q^j, \hat{E}(Q^j))\}$ as follows:
\begin{enumerate}
    \item Let $\{Q^j\}$ be the set of $Q \in \mathcal{Q}$ induced by allocations $q$ such that $q(t,\cdot) \in \{\rho^j\}$ for all $t \in T$.
    \item For any such $Q^j$, induced by $q^j$, let $\hat{E}(Q^j)$ be the set of $\lambda \in \Lambda$ such that
    \begin{equation*}
        \left(\frac{\lambda(u,t_u)}{\mu_u^{\bullet}(\tau)} \right)_{u\in U} \in \hat{e}(q^j(\cdot,t) )
    \end{equation*}
    for all $t \in T$.
\end{enumerate}
\end{lemma}

Clearly $\{(Q^j,\hat{E}(Q^j))\}$ constructed in this way is an $\alpha$-approximation cover. Of course, it is it is not necessary that an $\alpha$-approximation cover be generated in this way by an ex-post $\alpha$-approximation cover. However it is generally much easier to identify the ex-post variety. I show how simple algorithms which deliver approximate solutions to maximizing linear functions on $P^{XP}$ can be used to generate ex-post $\alpha$-approximation covers. 

\section{Exact characterizations: polymatroid constraints}\label{sec:exact_polymatroid}

The classic Border's Theorem is nice precisely because the extreme points of $P^{XP}$ are easily characterized by a greedy algorithm, and the equivalence sets are simple and large.

For now, consider only upper bounds: set $L(A) = 0$ for all $A \subset I$. (A similar argument is used to incorporate lower bounds.)

To characterize the extreme points of $\mathcal{I}$, for every $t$ we just need to solve problems of the form
\begin{equation}\label{eq:ex_post_assignment}
  \max_{q\geq 0}  \sum_{i \in I} \frac{\lambda(u,t_u)}{\mu_u^{\bullet}(t_u)} q(u,t)  \quad s.t. \sum_{u\in A} q(u,t) \leq C(A,t) \quad  \forall \ A \subset I
\end{equation}

The complexity of the characterization of interim realizability is determined by the complexity of the solution to this problem. This complexity is determined by the nature of the function $C$.

One very simple instance is the case we only need to know two coarse statistics about $\lambda(u,\tau)$
\begin{enumerate}
    \item The set $A$ of $(u,\tau)$ such that $\lambda(u,\tau) \geq 0$.
    \item The order on $\{(u,\tau) : \lambda(u,\tau) \geq 0\}$ induced by $(u,\tau) \mapsto \frac{\lambda(u,\tau)}{\mu_u^{\bullet}(\tau)}$. That is, the ordering of $A = \{(i_1,\tau_1),(i_2,\tau_2),\dots,(i_K,\tau_K)\}$ such that $k \mapsto \frac{\lambda(u_k,\tau_{k})}{\mu_{i_k}^{\bullet}(\tau_{k})}$ is decreasing.
\end{enumerate}
Say that any $\lambda,\lambda'$ which are equivalent in terms of properties $1.$ and $2.$ are \textit{ordinally equivalent}. (There may be multiple orders consistent with each $\lambda$; we just require a non-empty intersection for property $2.$) 

\begin{definition}
Say that $C$ is \textit{ordinally simple} if the solution to \cref{eq:ex_post_assignment} is the same for any ordinally equivalent $\lambda, \lambda'$.
\end{definition}

\begin{definition}
Say that a function $f$ from $\Lambda$ to the space of allocation rules is \textit{ordinally simple} if it produces the same allocation rule for any ordinally equivalent $\lambda,\lambda'$. 
\end{definition}

One ordinally simple function is defined by a greedy algorithm. Order $U$ in decreasing order of $\frac{\lambda(u,t_u)}{\mu_u^{\bullet}(t_u)}$. The greedy algorithm proceeds as follows
\begin{itemize}
    \item Set $i_1 = C(\{i_1\},t)$
    \item For $k > 1$, set $i_k = C(\{i_1, \dots, i_k \},t) - C(\{i_1, \dots, i_{k-1}\},t)$. 
\end{itemize}

When $C$ is submodular, the greedy algorithm solves the program in \cref{eq:ex_post_assignment}.

\begin{proposition}[\cite{dunstan1973greedy}]\label{thm:submodularC}
The greedy algorithm solves \cref{eq:ex_post_assignment} for every $\lambda$ if and only if $C(\cdot,t)$ is submodular for all $t$. 
\end{proposition}

A set $P^{XP} := \{ \rho \in \mathbb{R}_+^{U} : 0 \leq \sum_{u\in A} \rho(u) \leq C(A) \ \ \forall \ A \subset U \}$ is called a \textit{polymatroid} if and only if $C$ is submodular. 

\begin{observation}
The constraint $C(A) = 1$ for all $A$ defines the ex-post polytope for the classic interim allocation setting of \cite{border2007reduced}. This constraint is submodular. 
\end{observation}

Using the framework of \Cref{sec:geometry}, \Cref{thm:submodularC} tells us that if $C$ is submodular, the set of extreme points of $P^{XP}$ can be generated by varying the order on $U$ and the cut-off $k$, and applying the greedy algorithm.

Given $A \subset \mathcal{T}^*$ and an order $R$, let $\lambda \in \Lambda$ be such that \textit{a)} $\lambda$ is monotone with respect to order $R$, and \textit{b)} $\lambda(u,\tau) \geq 0$ iff $(u,\tau) \in A$. Let $\Lambda^{(A,R)}$ be the set of all such $\lambda$. Then the greedy algorithm applied to $\lambda \in \Lambda^{(A,R)}$ produces an allocation $q^{(A,R)}$, which I call a \textit{truncated greedy allocation}, and induces an interim allocation $Q^{(A,R)}$, which I call a \textit{truncated greedy interim allocation}. Let $\mathcal{Q}$ be the set of all truncated greedy interim allocations, found by varying the order $R$ and set $A$. Then \Cref{thm:submodularC} implies the following.

\begin{lemma}
If $P^{XP}(t)$ is a polymatroid (equivalently, $C(\cdot,t)$ is submodular) for all $t$, then $ext(\mathcal{I}) = \mathcal{Q}$. 
\end{lemma}

We also know that for each extreme point $\rho^*$ of $P^{XP}(t)$, the ex-post equivalence set $e(\rho^*)$ is a set of ordinally equivalent vectors. It is well known that the extreme points of a set of ordinally equivalent vectors are easily described (a proof is included for completeness).

\begin{lemma}\label{lem:decreasingextremepoints}
Fix any order $R$ on $U$. Let $\Gamma^{R}$ be the set of functions $\gamma \in \Gamma$ that are monotone with respect to $R$. Then the extreme points of $\Gamma^{R}$ are step functions taking values in $\{-1,1\}$. 
\end{lemma}
\begin{proof}
Given an order $R$ on $U$, to find the extreme points of the set of decreasing functions on $U$, we solve
\begin{equation*}
    \max_{z \in \R^U} \sum_{k = 1}^{U} z(i_k)y(u_k) \quad s.t. \  y \in [-1,1]^U, \ y(u_k) - y(u_{k+1}) \geq 0 \ \forall \ k < U. 
\end{equation*}
We can re-write the objective as
\begin{equation*}
   y(u_U) \sum_{k = 1}^{U} z(i_k) + \sum_{k = 1}^{U-1}\big( y(u_k) - y(u_{k+1}) \big)\sum_{\ell = 1}^k z(i_k)
\end{equation*}
The constraint that $y$ is decreasing and takes values in $[-1,1]$ is just the same as requiring that \textit{i)} $y(u_k) - y(u_{k+1}) \geq 0$, \textit{ii)} $\sum_{k = 1}^{U-1} y(u_k) - y(u_{k+1}) \leq 2$, and \textit{iii)} $y(u_U) \geq -1$. The the solution is clearly to find the $k$ that maximizes $\sum_{\ell = 1}^k z(i_k)$ and set $y(u_j) = 1$ for all $j \leq k$ and $y(u_j) = -1$ for all $j \geq k$ (or set $y = 0$ if $\sum_{\ell = 1}^k z(i_k) \leq 0$ for all $k$. 
\end{proof}

We can now translating \Cref{lem:decreasingextremepoints} into the extreme points of $E(Q^*)$. Since the greedy algorithm sets $q(u,t_u) = 0$ whenever $\lambda(u,t_u) <0$, we can restrict attention to non-negative $\lambda$. Then for each non-negativity set $A$, we only need to consider $\lambda^A$ defined by 
\begin{equation}\label{eq:lambdaR}
\lambda^A(u,\tau) =
    \begin{cases}
    \frac{1}{\mu_u^{\bullet}(\tau)} \quad &\text{if } (u,\tau) \in A \\
    0 \quad &\text{otherwise}
    \end{cases}
\end{equation}
Given an order $R$ on $\mathcal{T}^*$, which we index in $R$-decreasing order, let $\Lambda^R_+$ be the set of all $\lambda^A$ such that $A = \{(u_1,\tau_1),\dots, (u_k,\tau_k)\}$ for some $k \leq |\mathcal{T}^*|$. In other words, $\Lambda^R_+$ is the set of $\lambda$ such that if we index $\mathcal{T}^*$ in $R$-decreasing order, $k \mapsto \frac{\lambda(u_k,\tau_k)}{\mu_{u_k}^{\bullet}(\tau_k)}$ is a decreasing step function taking values in $\{0,1\}$.

\begin{lemma}\label{lem:submodular_cover}
For each $Q^{(A,R)} \in \mathcal{Q}$, let $E(Q^{(A,R)}) = co(\Lambda^R_+)$. Then $\{E(Q^*)\}_{Q^* \in \mathcal{Q}}$ is an equivalence cover. 
\end{lemma}

By \Cref{lem:equivalencecover_char} and \Cref{lem:submodular_cover}, we know that $Q\in \mathcal{I}$ iff for any ordering $R$ of $\mathcal{T}^*$, indexed in $R$ decreasing order, and any $A = \{(u_1,\tau_1),\dots, (u_k,\tau_k)\}$,
\begin{equation*}
    \lambda^A(Q) \leq \lambda^A(Q^{(A,R)}).
\end{equation*} 
All that remains is some housekeeping, using the fact that $(u,\tau) \rightarrow \frac{\lambda^A(u,\tau)}{\mu_{u}^{\bullet}(\tau)}$ is a step function taking values in $\{0,1\}$, and the properties of the greedy algorithm. This gives us the following characterization, which generalizes that of \cite{che2013generalized} (for the case of upper-bounds only) by allowing for state-dependent constraints. 

\begin{theorem}\label{thm:exact_submodular}
$Q \in \mathcal{I}$ if and only if 
\begin{equation*}
    \sum_{(u,\tau) \in A}Q(u,\tau)\mu_u^{\bullet}(\tau) \leq \sum_{t\in T} \mu(t) C(S(t,A),t)
\end{equation*}
for all $A \subset \mathcal{T}^*$.\footnote{To incorporate a non-trivial lower bound $L$, we can use a similar argument.}
\end{theorem}
\begin{proof}
Here is the housekeeping. For any $A$, any order on $A$, and any $t$, the greedy algorithm applied to this order produces an allocation $q^A$ such that
\begin{equation*}
   \sum_{i \in S(t,A)} q^A(i,t_u) = C(S(t,A),t)
\end{equation*}
and so 
\begin{align*}
\max_{Q' \in \mathcal{I}}\lambda^A(Q') &=  \sum_{t \in T} \mu(t) \sum_{i \in I} \frac{\lambda^A(i,t_u)}{\mu_u^{\bullet}(t_u)} q^A(t,i)\\
&= \sum_{i \in S(t,A)} q^A(i,t_u)\\
&= \sum_{t\in T} \mu(t) C(S(t,A),t)
\end{align*}
And for any $Q$
\begin{align*}
\lambda^A(Q) &= \sum_{(u,\tau) \in \mathcal{T}^*} \lambda^A(u,\tau)Q(u,\tau) \\
&= \sum_{(u,\tau) \in \mathcal{T}^*} \frac{\lambda^A(u,\tau)}{\mu_u^{\bullet}(\tau)} Q(u,\tau) \mu_u^{\bullet}(\tau) \\
&= \sum_{(u,\tau) \in A} Q(u,\tau)\mu_u^{\bullet}(\tau)
\end{align*}
\end{proof}

As a bonus, the greedy algorithm gives us an additional insight into the structure of $\mathcal{I}$. 

\begin{corollary}\label{cor:R_fosd}
$Q \in \mathcal{I}$ iff for any order on $R$ on $\mathcal{T}^*$ and any $j \leq |\mathcal{T}^*|$
\begin{equation*}
    \sum_{k = 1}^j Q(u_k,\tau_k)\mu_{u_k}^{\bullet}(\tau_k) \leq \sum_{k = 1}^j Q^{(\mathcal{T}^*,R)}(u_k,\tau_k)\mu_{u_k}^{\bullet}(\tau_k),
\end{equation*}
Where $\mathcal{T}^*$ is indexed in $R$-decreasing order. 
\end{corollary}

In other words, for an order $R$ on $\mathcal{T}^*$, $Q \in \mathcal{I}$ if and only if $Q$ is $R$-FOSD by $Q^{(\mathcal{T}^*,R)}$, the un-truncated interim greedy allocation for order $R$.\footnote{By $R$-FOSD I mean precisely the property defined in \Cref{cor:R_fosd}: that $Q$ assigns less ``ex-ante weight'' to upper-sets in the $R$ order then does the greedy algorithm applied to order $R$.} This characterization is related to the characterization of interim realizability via a majorization relationship in \cite{kleiner2021extreme}.

The results \Cref{sec:exact_outline} suggest that we can allow us to move beyond settings with sub/supermodular constraints by finding other ``simple'' algorithms that solve the ex-post maximization problem, when the greedy algorithm fails. 

\section{Approximate characterization: matching}\label{sec:multi_item}

The design of optimal mechanisms for allocating a single indivisible good among multiple agents has been an object of intensive study. Much less progress has been made on the allocation of multiple goods, despite the practical importance of such problems.\footnote{Examples include school choice \citep{abdulkadirouglu2003school}, the assignment of teachers to schools \citep{combe2018design}, and the assignment of police officers to districts \citep{ba2021police}.} The lack of progress is due to a number of technical challenges. For one, the characterization of implementability with multi-dimensional types, for example via cyclic monotonicity \citep{rochet1987necessary}, is much less tractable than the corresponding conditions for a single-item problem with one-dimensional types, in which in the allocation rule as a function of the type is often necessary and sufficient for incentive compatibility. Moreover, even in a complete information setting, if the principal's objective is non-linear in the allocation (for example due to ambiguity aversion or a preference for equity) the problem of finding the optimal mechanism may be computationally challenging.  

Perhaps a more fundamental concern, and the primary focus of this paper, is that even characterizing ``technological feasibility'' with multiple items is difficult. With a single item, the standard approach is to characterize the set of interim allocations that can be realized by some allocation rule, following \cite{border1991implementation}, and optimize directly over this space.\footnote{In brief, with a single item an \textit{allocation} specifies the probability with which each agent gets the item as a function of the type profile. An \textit{interim allocation} specifies, for each agent, the probability of getting the item as a function of only the agent's \textit{own} type. An allocation induces an interim allocation by, for each agent, taking the expected allocation probability across other agents' types. Realizability concerns the converse: given a candidate interim allocation,  how do we know that it can in fact be induced by some allocation.} If agents valuations are additive across items, this approach can be directly extended to the multi-item setting, as in \cite{cai2018constructive}. However this approach fails when agents have capacity constraints. Indeed, for the simple case in which agents have unit demand, \cite{gopalan2018public} show that no computationally tractable characterization of realizable interim allocations exists. Even abandoning computational considerations, theoretically meaningful characterizations of interim reliability with multiple items remain elusive. In contemporaneous work, \cite{zheng2022reduced} extends results of \cite{che2013generalized} to a multi-item setting, but the results do not apply when agents unit demand, or more general capacity constraints. \cite{lang2022reduced} provide an alternative characterization, but this is applicable to the unit-demand setting only when each agent has no more than two possible types, an especially strong restriction when there are multiple items. 
 
These challenges hinder the typical mechanism design approach to an allocation problem: choose a mechanism to maximize some objective, subject to the relevant incentive and feasibility constraints. As a result, much of the literature on allocation problems when agents have capacity constraints has adopted what \cite{budish2012matching} terms the ``good properties'' approach: specify a set of properties (strategy-proofness, stability, efficiency) which are desirable in a mechanism, and the construct a tractable procedure that delivers these properties. Seminal contributions in this line of work were made by \cite{gale1962college} (deferred acceptance) and \cite{shapley1974cores} (top-trading cycles). 

The good properties approach has been a source of theoretical insights and practical innovations. One drawback, however, is it is not always easy to map the designer's preferences over allocations to properties which should be imposed on the mechanism. It is not clear, for example, which properties should be imposed to capture a desire on the part of the designer for equity among agents. Moreover, commitment to a set of properties may imply a restriction to mechanisms that are arguably very far from optimal with respect to the designers true objective. Without a good understanding of what can be achieved by the mechanism design approach, it is difficult to even understand how far from optimal ``good properties'' mechanisms may be. 

The objectives of this section are threefold. The first is to contribute to the understating of interim realizability with multiple items. Second, to use this understanding to facilitate a ``mechanism design'' approach to the allocation problem, with a focus on accommodating non-linear objectives. Finally, to use the structure of interim realizability to attempt to bridge the gap between the good-properties and mechanism-design approaches. 

\subsection{Model}

There is a set $N$ of items and as set $I$ of agents (when it will not cause confusion, I also use $I$ and $N$ for the number of agents and items). Each unit $u \in U$ consists of an agent-item pair $(i,n)$. Each agent $i$ has a type drawn from finite set $T_i$, with typical element $\tau$. Let $T = T_1\times\dots\times T_N$ be the set of type profiles, with typical element $t$, where $t_i$ denotes $i$'s type in profile $t$. I refer to a realized type profile as a \textit{state}. Types are distributed according to the probability measure $\mu$ on $T$.\footnote{For the purposes of characterizing realizable interim allocations, it is not necessary to assume that types of independent, although this will of course play a role when discussing incentives. With a single item, type independence can also be used to simplify the characterization of interim realizability, as in \cite{border2007reduced}. The extent to which this is possible with multiple items remains an open question.} Let $\mu_i(t^{-i}|\tau)$ be the conditional distribution on $T^{-i}$ given $t_i = \tau$, and let $\mu_i^{\bullet}(\cdot)$ be the marginal distribution over $T_i$. I use the notation $t \sim (i,\tau)$ to denote that $t_i = \tau$.

In general, each agent $i$ has a capacity constraint $b^I_i$, which is the maximum number of items they can be allocated. Each item $n$ has a capacity constraint $b^N_n$, which is the maximum number of agents who can receive this item. In what follows, I restrict attention to the case of $b^I_i = b^N_n = 1$ for all $i\in I, n\in N$. For many of the results this is without loss of generality.\footnote{We can think of ``splitting up'' agents and items into multiple copies, such that each copy has capacity $1$.} 

The one-to-one matching constraint determines a (state-invariant) ex-post polytope $P^{XP}$ which is the set of $\rho : I\times N \rightarrow \mathbbm{R}_+$ such that 
\begin{enumerate}
    \item $\sum_{n \in N} \rho(i,n) \leq 1$ for all $i \in I$ \ \ (unit demand for agents).
    \item $\sum_{i \in I} \rho(i,n) \leq 1$ for all $n \in  N$ \ \ (unit supply of items).
\end{enumerate} 
These restrictions can have two interpretations. First, we could consider a setting with infinitely divisible items. In this case $q(i,n,t)$ is quantity of item $n$ that goes to $i$ in state $t$. 

Alternatively, we can consider settings with indivisible items. In this case $q(i,n,t)$ is the probability that item $n$ goes to $i$ in state $t$. In other words, $q$ is the marginal of the joint distribution the designer induces over assignments of items to agents. The Birkhoff-von Neumann Theorem tells us exactly that $q$ can be the marginals of such a distribution if and only if $q$ satisfies the unit demand and unit supply conditions above. Note that by focusing on $q$ defined in this way, we are implicitly assuming that only the marginals of the joint distribution over assignments matter, both for the agents and the designer. For details and further discussion, see \Cref{sec:birkhoff}. 

An allocation $q$ induces an \textit{interim allocation} $Q$ where $Q: \mathcal{T}^* \mapsto [0,1]$ is defined by
\begin{equation}\label{eq:interim_def}
    Q(i,\tau,n) := E_{t_{-i}\sim \mu_i(\cdot|\tau)}[q(t_{-i},\tau,i,n)] = \sum_{t_{-i} \in T_{-i}} q(t_{-i},\tau,i,n) \mu_i(t_{-i}|\tau).
\end{equation}
In other words, $Q(i,\tau,n)$ is the probability that $i$ gets item $n$, conditional on having type $\tau$. 

Because an interim allocation $Q$ is a object of significantly lower dimension than the allocation $q$, it is often convenient to work directly with the interim allocation. Moreover, as long as agents' payoffs are linear in their allocation, the interim allocation is all that is relevant for the agent's incentives. This interim approach was first used for single-item problems by \cite{maskin1984optimal}. In order to work directly with the interim allocation however, it is necessary to first characterize the set of valid interim allocations, i.e. those that are induced by some allocation as in \cref{eq:interim_def}. I call a candidate interim allocation $Q: \mathcal{T}^*$ \textit{realizable} if there exists an allocation $q$ such that \cref{eq:interim_def} holds for all $i,\tau,n$.

For the single-item case, this characterization was provided by \cite{border1991implementation}, who proved a conjecture of \cite{matthews1984implementability}. This characterization was expanded on by \cite{border2007reduced}, \cite{mierendorff2011asymmetric} and \cite{che2013generalized}, among others. However none of these characterizations applies to the one-to-one assignment problem. 

\subsection{An approximate characterization for matching}\label{sec:mainresults}

The results of \Cref{sec:exact_polymatroid} do not apply to the matching problem because $P^{XP}$ is not a polymatroid.\footnote{In fact, $P^{XP}$ is the intersection of two polymatroids, one generated by the ``row'' constraints, and one by the ``column'' constraints.} To see this, define For any set $X \subseteq I\times N$, define 
\begin{equation*}
    C(X) = \max_{\rho \in P^{XP})} \sum_{(i,n) \in X} \rho(i,n). 
\end{equation*}

\begin{observation}
The constraint $C$ is not submodular. This is easily illustrated in the following example with two agents $\{i_1,i_2\}$ and two items $\{n_1,n_2\}$. Let $a = \{(i_1,n_1), (i_1,n_2)\}$ and $b = \{(i_1,n_1), (i_2,n_1)\}$. Then $C(a) = C(b) = 1$. However $C(a\cap b) = C(\{(i_1,n_1)\}) = 1$ and $C(a\cup b) = C(\{(i_1,n_1), (i_1,n_2), (i_2,n_1)\}) = 2$, violating submodularity. 
\end{observation}

As a result, the extreme points of $P_C$ are not characterized by the greedy algorithm. More generally, the task of identifying an equivalence cover is complicated by the fact that the solution to the program
\begin{equation*}
    \max_{\rho \in P^{XP}} \sum_{(i,n,)} \rho(i,n) \gamma(i,n)
\end{equation*}
is sensitive to more features of $\gamma$ than just the order it induces on $I\times N$. 

It turns out, however, that a modified algorithm, the \textit{greedy matching algorithm} is generates an $\alpha$-approximation cover. 

Let $R$ be an strict ordering of $\mathcal{T}^*$. Index the elements of $\mathcal{T}^*$ in $R$-decreasing order, i.e. such that $(i_k,\tau_k,n_k)$ precedes $(i_{k+1},\tau_{k+1},n_{k+1})$. Define an \textit{$R$-greedy-matching allocation} $q^R$ using the following recursive greedy algorithm. Let $q^R(t,i_1,n_1) = 1$ for all $t\sim (i_1,\tau_1)$. For each $k >1$, let $q^R(t,i_k,n_k) = 1$ if and only if
\begin{itemize}
    \item $t \sim (i_k,\tau_k)$, and
    \item For every $l < k$ such that $t \sim (i_l,\tau_l)$, both $i_k \neq i_l$ and $n_k \neq n_l$.
\end{itemize}
In words, the $R$-greedy-matching allocation assigns as much weight as possible to $(i_1,\tau_1)$ getting item $n_1$, and then for subsequent $(i_k,\tau_k,n_k)$, assigns as much weight as possible to $(i_k,\tau_k)$ getting item $n_k$, conditional on the previous assignments. By construction the $R$-greedy allocation is realizable, for any $R$. 

Given the $R$-greedy allocation $q^R$, let $Q^R$ be the induced  \textit{$R$-greedy-matching interim allocation}, defined by
\begin{equation*}
    Q^R(i,\tau,n) = \sum_{t_{-1} \in T^{-i}} q^R((t_{-i},\tau),i,n)\mu_i(t^{-i}|\tau). 
\end{equation*}

An allocation $q$ is called a \textit{truncated greedy-matching allocation} if there exists an order $R$ on $\mathcal{T}^*$ and an number $k^*$ such that $q(t,i_k,\tau_k) = q^R(t,i_k,\tau_k)$ for all $t \sim (i_k,\tau_k)$ and all $k \leq k^*$, and $q(t,i_k,\tau_k) = 0$ for all $t \sim (i_k,\tau_k)$ and all $k > k^*$. The interim allocation $Q$ induced by such an truncated greedy matching allocation satisfies $Q(i_k,\tau_k,n_k) = Q^R(i_k,\tau_k,n_k)$ for all $k \leq k^*$, and $Q(i_k,\tau_k,n_k) = 0$ for all $k > k^*$. A write $Q^{(k,R)}$ for a truncated greedy matching allocation which is truncated at $k$.

Let $\mathcal{Q}^M$ be the set of interim allocations induced by truncated greedy allocations. It is helpful also to define the set of weighted greedy interim allocations 
\begin{equation*}
    \mathcal{X} := \left\{X \in \mathbb{R}^{(\mathcal{T}^*)} : X(i,\tau,n) = Q(i,\tau,n)\mu_i^{\bullet}(\tau) \text{ for some } Q \in \mathcal{Q}\right\}.
\end{equation*}

Say that the interim allocation $Q$ is $R$-FOSD dominated by $\hat{Q}$ if
\begin{equation*}
       \sum_{l=1}^k Q(i_l,\tau_l,n_l)\mu_{i_l}^{\bullet}(\tau_l) \leq \sum_{l=1}^k \hat{Q}(i_l,\tau_l,n_l)\mu_{i_l}^{\bullet}(\tau_l)
\end{equation*}
for all $k$, where $\mathcal{T}^*$ is indexed in $R$-decreasing order. 

Recall the definition of $\Lambda^R_+$ defined by \Cref{eq:lambdaR}: $\Lambda^R_+$ is the set of $\lambda$ such that if we index $\mathcal{T}^*$ in $R$-decreasing order, $k \mapsto \frac{\lambda(i_k,n_k,\tau_k)}{\mu_{i_k}^{\bullet}(\tau_k)}$ is a decreasing step function taking values in $\{0,1\}$.

Denote by $\frac{1}{2}Q$ the half-half mixture between $Q$ and the null interim allocation that does not assign any items. In other words, $\frac{1}{2}Q$ scales down all assignment probabilities by $1/2$.

\begin{theorem}\label{thm:half_char}
If $Q \in \mathcal{I}$ then $Q$ satisfies 
\begin{equation}\tag{BM}\label{eq:BM}
    \sum_{(i,\tau,n)\in A} Q(i,\tau,n)\mu_i^{\bullet}(\tau) \leq \sum_{t\in T} \mu(t) C(S(t,A)).
\end{equation}
Conversely, if $Q$ satisfies (\ref{eq:BM}) then
\begin{enumerate}[i.]
    \item For any order $R$ on $\mathcal{T}^*$, $\frac{1}{2} Q$ is $R$-FOSD dominated by the greedy interim allocation $Q^R$. \label{thm:1:part:i}
    \item $\frac{1}{2} Q \in co(\mathcal{Q}^M)$. That is, $\frac{1}{2}Q$ can be written as a convex combination of interim allocations induced by truncated greedy-matching allocations. \label{thm:1:part:ii}
    \item $\frac{1}{2} Q$ is realisable. \label{thm:1:part:iii}
\end{enumerate}
\end{theorem}

\Cref{thm:half_char} part \ref{thm:1:part:iii} follows immediately from part \ref{thm:1:part:ii}. Part \ref{thm:1:part:i} states that if $Q$ satisfies the necessary (\ref{eq:BM}) condition then scaling $Q$ by $\frac{1}{2}$ defines an interim allocation which assigns less weight to the upper-level-sets of any order $R$ on $\mathcal{T}^*\times N$ than the interim allocation generated by greedy matching algorithm applied to that order. 

\Cref{thm:half_char} is proven via the following intermediate result. 
\begin{lemma}\label{lem:matching_approximationcover}
For $Q^{(k,R)} \in \mathcal{Q}^M$, let $\hat{E}(Q^{(k,R)}) = co(\Lambda^R_+)$. Then $\{Q,\hat{E}(Q)\}_{Q \in \mathcal{Q}^M }$ is a $1/2$-approximation cover. 
\end{lemma}

\subsubsection{Discussion of \texorpdfstring{\Cref{thm:half_char}}{}}

There are multiple ways in which \Cref{thm:half_char} is useful for understanding optimal mechanisms. Assume that the principal's objective is (weakly) concave in $Q$ and normalized so that the payoff from the null allocation (no agents get any good for any type profile) is zero. Then scaling $Q$ by $\frac{1}{2}$ reduces the principal's payoff by less than $\frac{1}{2}$ (in other words, gives a $2$-approximation to the optimal mechanism). One could then maximize the principal's objective over interim allocations subject to (\ref{eq:BM}), as well as the relevant IC constraints, then scale by $\frac{1}{2}$. By \Cref{thm:half_char} this guarantees at least $\frac{1}{2}$ maximum payoff. Moreover, if agents maximize expected utility then scaling the allocation (and payments when applicable) preserves incentive compatibility. 

An alternative is to maximize directly over the set $co(\mathcal{Q})$. Again, \Cref{thm:half_char} guarantees that this yields at least half of the maximum payoff. This maximization problem may be intractable, given that \Cref{thm:half_char} does not provide a tractable way of checking whether a given interim allocation is in fact in $co(\mathcal{Q})$. However in many settings it suffices to solve a simpler problem. Assume there exists an order $R$ on $\mathcal{T}^*$ such that the principal's payoff is increasing in $R$-FOSD shifts. Since $Q^R$ $R$-FOSD dominates $\frac{1}{2}Q$, $Q^R$ delivers at least $\frac{1}{2}$ of the payoff from any \textit{realizable} interim allocation, regardless of whether such allocations are incentive compatible. It remains to check if $Q^R$ is in fact IC. This boils down to checking whether $R$ satisfies certain conditions. In \Cref{sec:applications} I study settings in which natural conditions on the principal's preferences guarantee that $Q^R$ is IC. The structure of greedy allocations simplifies this exercise. 

\subsubsection{\texorpdfstring{\Cref{thm:half_char}}{} proof}

In order to prove \Cref{thm:half_char} it is necessary to introduce some new objects. Given a pair $(i,n)$, let $\Gamma(i,n)$ be the union of the row and column to which $(i,n)$ belongs. That is,
\begin{equation*}
    \Gamma(i,n) := \{j,n\}_{j\in I} \cup \{i,m\}_{m\in N}.
\end{equation*}
For $\rho \in \mathcal{D}$ and $a \subseteq I\times N$, let $\Gamma^*(\rho,a)$ be \textbf{the projection of $\rho$ in $a$}, defined by 
\begin{equation*}
    \Gamma^*(\rho,a) = \cup_{\{(i,n) \in a : \rho(i,n) =1\}} \Gamma(i,n). 
\end{equation*}
Say that  the projection of $\rho$ in $a$ \textbf{covers} $a$ if $a \subseteq \Gamma^*(\rho,a)$.

\begin{lemma}\label{lem:halfbound}
For any assignment $\rho \in \mathcal{D}$ such that the projection of $\rho$ in $A$ covers $A$
\begin{equation*}
    \sum_{(i,n)\in A} d(i,n) \geq \dfrac{1}{2}c(A).
\end{equation*}
\end{lemma}
\begin{proof}
We have $c(A) \leq \max_{d \in \Delta(D)} \sum_{(i,n) \in \Gamma^*(\rho,A)} d(i,n) \leq 2 \sum_{(i,n)\in A} \rho(i,n)$. The first inequality follows from the fact that the projection of $\rho$ covers $A$. The second follows from the fact that the total mass in each row and column can be no more than 1, and $\Gamma^*(\rho)$ is the union of $\sum_{(i,n)\in A} \rho(i,n)$ rows and $\sum_{(i,n)\in A} \rho(i,n)$ columns.  
\end{proof}

\begin{proof}[Proof of \Cref{thm:half_char}]
\Cref{lem:matching_approximationcover} is immediate from \Cref{lem:halfbound} and \Cref{lem:expost_approximationcover}. Then \Cref{thm:half_char} is implied by \Cref{theorem:approximation_cover_char}. For clarity, I also provide a self-contained proof of \Cref{thm:half_char} in \Cref{app:half_char}. 
\end{proof}

\subsection{Tightening the bound}

Scaling the interim allocation $Q$ by $\frac{1}{2}$, i.e. mixing equally between $Q$ and the null allocation, causes the principal to leave each object unassigned with probability at least $1/2$. This may be a very undesirable outcome. Thus the fact that optimizing over $co(\mathcal{Q})$ produces a mechanism that is better that $\frac{1}{2}Q^*$, where $Q^*$ is the fully optimal mechanism, may not be enough to recommend this approach. However in practice it may be possible to strengthen the payoff approximation by exploiting the ``slack'' introduced by scaling down the interim allocation. 

The key observation is that the 2-approximation to the payoff of $\frac{1}{2}Q$ can be obtained by using the greedy algorithm to only part of $\mathcal{T}^* \times N$. This is an implication of the following useful observation.\footnote{A similar observation is made in \cite{anshelevich2016truthful}, Lemma B.7, of which \Cref{lem:partialfosd} is a generalization.} Given a distribution $F$ on $[0,1]$, let $p_F(\alpha) := \inf\{x : F([0,x]) \geq \alpha \}$. 
\begin{lemma}\label{lem:partialfosd}
Let $F$ and $G$ be measures on $[0,1]$. Assume
\begin{enumerate}[i.]
    \item $F([0,z]) \geq G([0,z])$ for all $z \in [0,1]$, and
    \item $1 = F([0,1]) = \frac{1}{\alpha}G([0,1])$ for some $\alpha \in (0,1]$. 
\end{enumerate}
Define the measure $H$ by $H([0,z]) = F([0,z])$ for $z \leq q_F(\alpha)$ and $H((q_F(\alpha),1]) = 0$. Then $H([0,z]) \geq G([0,z])$ for all $z \in [0,1]$. 
\end{lemma}
\begin{proof}
Immediate from condition $i$ for $z \leq p_F(\alpha)$, and from condition $ii$ for $z > p_F(\alpha)$.
\end{proof}

\Cref{lem:partialfosd} has the following relationship with \Cref{thm:half_char}. Assume the principal's objective is concave and monotone with respect to $R$-FOSD shifts for some order $R$ on $\mathcal{T}^*\times N$. If $Q$ satisfies the (\ref{eq:BM}) condition then \Cref{thm:half_char} says that $\frac{1}{2}Q$ is $R$-FOSD dominated by $Q^R$, the $R$-greedy interim allocation. Thus we have already concluded that $Q^R$ is a $2$-approximation for the principal's problem. \Cref{lem:partialfosd} allows us to say more: in order to obtain a realizable interim allocation $\hat{Q}$ that $R$-FOSD dominates $\frac{1}{2}Q$ it is not necessary to use the greedy algorithm to define $\hat{Q}$ on all of $\mathcal{T}^*\times N$. It is sufficient to find the $\alpha$ such that 
\begin{equation*}
   \frac{1}{2} \sum_{(i,\tau,n) \in \mathcal{T}\times N} Q(i,\tau,n)\mu_i^\bullet(\tau) = \alpha \sum_{(i,\tau,n) \in \mathcal{T}\times N} Q^R(i,\tau,n)\mu_i^\bullet(\tau) 
\end{equation*}
and the smallest $\ell$ such that 
\begin{equation*}
  \sum_{k =1}^{\ell} Q^R(i_k,\tau_k,n_k)\mu_{i_k}^\bullet(\tau_k)  \geq \alpha \sum_{(i,\tau,n) \in \mathcal{T}\times N} Q^R(i,\tau,n)\mu_i^\bullet(\tau) 
\end{equation*}
(where $\mathcal{T}^*\times N$ is indexed in $R$-decreasing order). The we can let $\hat{Q}(i_k,\tau_k,n_k) = Q^R(i_k,\tau_k,n_k)$ for $k \leq \ell$. This is sufficient to guarantee that $\hat{Q}$ that $R$-FOSD dominates $\frac{1}{2}Q$, and we can define $\hat{Q}$ in any way we want below $\ell$.  

If $Q$ is in fact the optimal interim allocation (and so of course realizable) then $\sum_{(i,\tau,n) \in \mathcal{T}\times N} Q(i,\tau,n)\mu_i^\bullet(\tau) \leq \min\{|I|,|N|\}$. Since any greedy allocation assigns every item if $|I| \geq |N|$, and gives every agent an item if $|N| \geq |I|$, we have $\sum_{(i,\tau,n) \in \mathcal{T}\times N} Q^R(i,\tau,n)\mu_i^\bullet(\tau) = \min\{|I|,|N|\}$. Thus $\alpha$ in the above derivation will be less than $1/2$. As a result, only half of the total ex-ante allocation weight needs to come from the greedy algorithm. In certain applications, this may be useful in deriving tighter worst-case bounds.


\section{Matching applications}\label{sec:applications}

\subsection{Principal's preferences}

This section illustrates a few of the forms which the principal's objective may take. We then proceed to applications. 

\vst
\noindent\textit{Utilitarian Welfare}

One commonly studied objective is utilitarian welfare: there exists a weight function $v \in \mathbbm{R}^{\mathcal{T}^*}$ such that the principal's payoff is given by
\begin{equation*}
    V(Q) = \sum_{(i,\tau,n) \in \mathcal{T}^*} Q(i,\tau,n) \mu_i^{\bullet}(\tau) v(i,\tau,n).
\end{equation*}
This form of the objective can also accommodate revenue maximization, using the usual virtual values transformation of \cite{myerson1981optimal}.

\vst
\noindent\textit{Equity preferences}

The utilitarian objective does not embody an explicit preference for equity. Such preferences can be captured by objectives that are concave in the interim allocation. One way to extend the utilitarian welfare objective to incorporate a preference for equity is via the following \textit{rank dependent} preferences. As before, let $v \in \mathbbm{R}^{\mathcal{T}^*}$ be a weight function. Label $\mathcal{T}^*$ such that $k \mapsto v(i_k,\tau_k,n_k)$ decreasing. Then for some increasing function $f: \R \rightarrow \R$ such that $f(0) = 0$, let
    \begin{equation*}
        V(Q) = \sum_{k = 1}^{|\mathcal{T}^*|-1} \big(v(i_k,\tau_k,n_k) - v(i_{k+1},\tau_{k+1},n_{k+1})\big)f\left(\sum_{j=1}^k Q(i,\tau,n) \mu_i^{\bullet}(\tau) \right)
    \end{equation*}
If $f$ is the identify function and the weight function is normalized such that $\min v(i,\tau,n) = 0$ then this is exactly the the utilitarian welfare objective with weight function $v$. However if $f$ is concave then so is $V$. In this case the principal wants to smooth the allocation.

\vst
\noindent\textit{Model uncertainty}

Uncertainty about model fundamentals, such as the ex-post payoff of allocations, can also induce concavity in the principal's objective. Again, the utilitarian model can be extended to accomodate these concerns. Let $W \subset \mathbbm{R}^{\mathcal{T}^*}$ be a set of weight functions which the principal entertains. The principal takes a cautious approach to their uncertainty about the model: they maximize agaist the wost case. The payoff is give by
    \begin{equation*}
        V(Q) = \min_{v \in W} \sum_{(i,\tau,n) \in \mathcal{T}^*} Q(i,\tau,n) \mu_i^{\bullet}(\tau) v(i,\tau,n).
    \end{equation*}

\vst
Throughout the applications I maintain the assumption that that the principal's objective is (weakly) concave in $Q$ and normalized so that the payoff from the null allocation (no agents get any good for any type profile) is zero. All of the above objectives satisfy these assumptions. Moreover, I assume that there exist an order $R$ on $\mathcal{T}^*$ such that the principal's preferences are monotone with respect to $R$-FOSD shifts. This is satisfied by utilitarian and rank-dependent welfare. It is also satisfied by the max-min model if every pair of weight functions in $W$ is comonotone.


\subsection{Relation to Deferred Acceptance}\label{sec:DA}

Consider a standard school-choice setting. Agents $I$ are students and items $N$ are schools. As is standard in this literature, let each student's type set $T_i$ be a set of strict (ordinal) rankings over schools. (I allow for type sets that do not include all such rankings). For $x,y\in N$, write $x\tau y$ to denote that $x$ is above $y$ in the ranking $\tau$. 

The (student proposing) deferred acceptance (DA) algorithm takes as inputs students' reported preferences over schools and schools' rankings over students, and produces the student optimal stable matching.  

The meaning of schools' rankings of students, which following the literature I refer to as \textit{priorities}, is a matter of some debate, and may vary depending on the setting. In many settings, such as the Boston school system, these priorities are generated by some centralized setting. The priorities may reflect whether or not a student is in the walk-zone for a given school, if they have siblings attending the school, their academic records, and various affirmative action policies (e.g \cite{hafalir2013effective}). 

If the priorities are set to reflect the preferences of the principal, in this case the school district, over allocations, it is natural to ask how well DA does. To be precise, how close to maximizing their objective can the principal come by choosing priority rankings over students for each school and running deferred acceptance? \Cref{thm:half_char} allows us to provide a partial answer to this question.

Let $\succ$ represent the principal's priority order on $\mathcal{T}^*$. Given the priority order, I define a notion of relevance of a given realization $(i,\tau,n)$.

\begin{definition}
Given the principal's priority order $\succ$, say that $(i,\tau_i,n)$ is \textbf{blocked} if there is no $t_{-i}$ such the greedy algorithm applied to order $\succ$ gives item $n$ to $i$ in state $(\tau_i,t_{-i})$. Say that $(i,\tau_i,n)$ is unblocked if it is not blocked. 
\end{definition}

\begin{definition}
Say that $\succ$ is \textbf{item-ranking consistent} if for all $n$ and $i\neq i'$, if $(i,\tau,n)$ is unblocked and $(i,\tau,n) \succ (i',\tau',n)$ then $(i,\tau'',n) \succ (i',\tau''',n)$ for all $\tau''$ and all $\tau'''$ such that $(i',\tau''',n)$ is unblocked. 
\end{definition}

The following simple observation says that item-ranking consistency makes it possible to back out item priority orders from $\succ$ in a consistent way.

\begin{lemma}\label{lem:item_rankings}
The following are equivalent
\begin{enumerate}[(i.)]
    \item $\succ$ is item-ranking consistent
    \item There exists a family $\{>_n\}_{n \in N}$ of orderings of $I$ such that for any $n$ and any unblocked $(i,\tau,n)$, if $(i,\tau,n) \succ (j,\tau',n)$ and $i\neq j$ then $i >_n j$.
\end{enumerate}
\end{lemma}
\begin{proof}
First, assume $\succ$ is item-ranking consistent. Fix $n$, and define $>_n$ as follows. Let $D \subset I$ be the set of agents such that for any $i\in D$, there exists $\tau$ such that $(i,\tau,n)$ is unblocked. Then for $i\in D$, let $i >_n i'$ if there exists $\tau'$ such that $(i,\tau,n) \succ (i',\tau',n)$. This is well defined under item-ranking consistency. The remaining agents $I\setminus D$ can be ordered in any way at the bottom of $>_n$, below all $i\in D$.

Conversely, suppose condition (ii) holds. If there is a violation of item-ranking consistency then there exists unblocked $(i,\tau, n)$ and $(i',\tau''',n)$ such that for some $\tau',\tau''$, $(i,\tau,n) \succ (i',\tau',n)$, and $(i',\tau''', n) \succ (i,\tau'',n)$. By property (ii), $(i,\tau,n) \succ (i',\tau',n)$ implies $i >_n i'$, and $(i',\tau''', n) \succ (i,\tau'',n)$ implies $i' >_n i$, which is absurd.
\end{proof}

\begin{definition}
Say that $\succ$ is \textbf{welfarist} if for any unblocked $(i,\tau,n)$, if  $(i,\tau,n) \succ (i,\tau,n')$ then $n\tau n'$.
\end{definition}

This property is straightforward: the priority order is welfarist if it respects the preferences of the agents. We require this only for unblocked tuples.

\begin{theorem}\label{thm:DA_guarantee}
Assume that the priority order is \textbf{welfarist} and \textbf{item-ranking consistent}. Then there exists a family of priority rankings for schools such that 
\begin{enumerate}
    \item DA is equivalent to the greedy algorithm applied to $\succ$ (so produces the student optimal stable matching).
    \item DA guarantees half the principal's full-information payoff. 
\end{enumerate}
\end{theorem}

\cite{roth1985college} shows that in the deferred acceptance mechanism it is a dominant strategy for proposers with unit demand to report truthfully. Thus \Cref{thm:DA_guarantee} has the following incentive implications. 

\begin{corollary}
If $\succ$ is welfarist and item-ranking consistent then the greedy algorithm applied to $\succ$ is DSIC. 
\end{corollary}

\begin{remark}
The set of welfarist and item-ranking consistent principal priority orders is non-empty. For example, assume that all item rankings are identical, given by $>$. Then order $\mathcal{T}^*$ ``lexicographically'': first, rank according to $>$, and then according to the individual's preferences. 
\end{remark}

In \Cref{app:more_on_DA} I explore in more detail the types of allocations generated by the greedy algorithm applied to a priority order that is welfarist and item-ranking consistent. Unsurprisingly, any serial dictatorship can be generated in such a way. However, I show that other types of allocations can also generated by such a priority.    

\vst
\noindent\textit{Proof of \Cref{thm:DA_guarantee}.} Let $\{>_n\}_{n\in N}$ be (one of) the item rankings defined by \Cref{lem:item_rankings}. I first show that the greedy algorithm applied to $\succ$ produces a stable matching for every state $t = \{\tau_i\}_{i\in I}$. Suppose that agent $j$ is allocated item $n$. Then $(j,\tau_j,n)$ is certainly not blocked. Suppose there exists an agent $i$ allocated item $n'$, such that $n \tau_i n'$. Since $\succ$ is welfarist, it must be that $(j,\tau_j,n) \succ (i,\tau_i,n) \succ (i,\tau_i,n')$. Since $\succ$ is item-ranking consistent, it must be that $j >_n i$. Thus the matching is stable. 

I now show that the matching is agent-optimal. Say that an agent $i$ is \textit{forestalled} from item $n$ at stage $k$ of the greedy algorithm if after stage $k-1$: \textit{i}) item $n$ has already been assigned, \textit{ii}) $i$ is unassigned, and \textit{iii}) $i$ prefers $n$ to any object not yet assigned. 

\textit{Claim 1.} If at stage $k$ of the greedy algorithm agent $i$ is assigned item $n$ and $i$ has not been forestalled from any object that $i$ could receive in some stable matching, then $n$ must be $i$'s most preferred item among those that $i$ can receive in any stable matching. 

\textit{Proof of Claim 1.} Let $n'$ be the object assigned to $i$, and let $k$ be the round of the greedy algorithm at which this assignment is made. Suppose $i$ prefers $n$, which $i$ could receive in some stable matching. Since $(i,\tau_i,n')$ is not blocked, it must be that $(i,\tau_i,n) \succ (i,\tau_i,n')$, since otherwise we would have a violation of welfarism of $\succ$. Since $i$ was not assigned $n$, it must be that $n$ was assigned at some stage of the algorithm before $k$. But then $i$ is forestalled from item $n$ at some stage, which we assumed was not the case. 

Given Claim 1, we want to show that at no stage is an agent forestalled from any object that $i$ could receive in some stable matching. The proof is by induction.\footnote{This proof is inspired by that used in \cite{gale1962college} to show that DA is proposer-optimal.}

Assume as the induction hypothesis that up to stage $k$ in the greedy algorithm, no agent has been forestalled from an object that they could receive in some stable matching. Let $(i,\tau_i,n)$ be the $k+1$ highest realization according to $\succ$. We want to show that at this stage $i$ is not forestalled from an object that they could get in another stable matching. Suppose towards a contradiction that $i$ is forestalled from item $n$, which $i$ could receive in some stable matching. This means that there is some unblocked $(j,\tau_j,n) \succ (i,\tau_i,n)$, and so since $\succ$ is item-ranking consistent, by \Cref{lem:item_rankings} it must be that $j >_n i$. By the induction hypothesis, $j$ has not been forestalled at any item that they could get in another stable matching. By Claim 1, this means that $j$ prefers $n$ to any other object that $j$ could get in a stable matching. Consider a stable matching in which $i$ receives $n$. Then $j$ receives some other item that they like less than $n$. Since we have already concluded that $j >_n i$, this matching cannot be stable. 

By induction, no agent is ever forestalled from an item that they could get in any other stable matching. By Claim 1, this means that all agent's receive the best item among those that they could get in any stable matching, i.e. the matching is agent-optimal. The equivalence to the deferred acceptance matching follows from \cite{gale1962college}. 
\hfill\qedsymbol


\subsection{Cardinal preferences and serial dictatorship}

In the previous application, only ordinal preferences of the agents were considered. In some settings, the principal may care about the intensity of agents' preferences. This is especially relevant if the principal is able to use transfers to illicit cardinal preferences. 

Suppose each agent's type is $\tau = (v,h)$, where $v \in \mathbbm{R}_+$ is \textit{vertical type}, representing preference intensity, and $h$ is a permutation of $h^1 > h^2 > \dots h^N$ representing ordinal preferences, referred to as the \textit{horizontal type}. The payoff of an agent with type $\tau = (v_{\tau},h_{\tau})$ who receives item $n$ with probability $q(n)$ and makes payment $p$ is given by
\begin{equation*}
    v_{\tau} \sum_{n \in N} q(n)h_{\tau}(n) - p,
\end{equation*}
or more compactly $v_{\tau} q\cdot h_{\tau} - p$. This model generalizes the ranked-item auction studied in \cite{kleiner2021extreme} in two ways. First, agents may have different ordinal rankings over the items. Second, the ordinal rankings may potentially be agents' private information.

Types are assumed to be independent across agents, with distribution $\mu_i^{\bullet}(v,h)$. Let $F_i(\cdot|h)$ be the CDF of $i$'s vertical type, with mass function $f_i(\cdot|h)$, conditional on horizontal type $h$. Let $\nu_i(v_{\tau},h_{\tau}) := v_{\tau} - \frac{1-F_i(v_{\tau}|h_{\tau})}{f_i(v_{\tau}|h_{\tau})}$. Suppose that horizontal types were known to the principal. Then the expected revenue from interim allocations $Q_i:T_i\times N \rightarrow [0,1]$ is 
\begin{equation*}
 Rev(Q) :=  \sum_{v \in V} \sum_{h \in H} \sum_{n\in N} \sum_{i\in I} \nu_i(v|h) \cdot h(n) \cdot f_i(v)  g_i(h) \cdot Q_i(v,h,n) 
\end{equation*}
Moreover, if horizontal types are known then an interim allocation $Q_i$ is implementable iff $v_{\tau}> v_{\tau'}  \Rightarrow h_{\tau} \cdot \Vec{Q}_i(\tau) \geq h_{\tau'} \cdot \vec{Q}_i(\tau')$. The principal's objective is a weighted sum of revenue and a welfare measure, given by 
\begin{equation*}
V(Q) = U(Q) + \beta Rev(Q).
\end{equation*} 

An allocation is called a \textit{type-specific serial dictatorship} if for each profile of vertical types, there exists a priority order on agents such that for any profile of horizontal types, the agents are awarded their highest-ranked item in order of priority. That is, the highest-priority agent receives their top-ranked item, and each subsequent agent in the priority order receives their top item among those not allocated to higher-priority agents. Importantly, while the priority order is fixed for each profile of vertical types, it may vary across vertical-type profiles. It is easy to see that if an order $R$ on $\mathcal{T}^*$ is welfarist then the $R$-greedy allocation is a type-specific serial dictatorship.

Assume that $v \mapsto \nu_i(v|h)$ is increasing for all $i\in I$ and $h \in H$. Then using the characterization of \Cref{thm:half_char} we can derive a revenue bound via serial dictatorship allocations. 

\begin{proposition}\label{prop:SD_knownh}
If horizontal types are observed, and $v \mapsto v_i(v,h)$ is monotone for all $i,h$ then the $R$-greedy allocation is a $2$-approximation and is DSIC type-specific serial dictatorship (with suitable payments). 
\end{proposition}

\begin{remark}
In fact, the $R$-greedy allocation is a 2-approximation to both parts of the principal's objective, welfare and revenue. 
\end{remark}

\begin{proof}
First, define an arbitrary order on agents, called the \textit{base priority order}. Order $\mathcal{T}^*$ in decreasing order of $\nu_i(v|h)h(n)$, i.e. such that $\nu_{i_k}(v_{k}
|h_k)h_k(n_k) \geq \nu_{k+1}(v_{k+1}|h_k)h_{k+1}(n_{k+1})$ for all $k$, with ties broken according to the base priority order. Call this the \textit{design order}. Let $y_k = \nu_{i_k}(v_{k}|h_k)h_k(n_k)$. I first show that the greedy algorithm applied to this order delivers higher revenue that $\frac{1}{2} Q$ for any interim allocation $Q$ that satisfies condition (\ref{eq:BM}). Since (\ref{eq:BM}) is necessary for realizability, this implies that the greedy algorithm delivers at least half the revenue that could be obtained with known types. The revenue from $\frac{1}{2}Q$ is given by
\begin{equation*}
   \frac{1}{2} \sum_{k=1}^K y_k \mu_i^{\bullet}(\tau) Q(i,\tau,n)
\end{equation*}
so the fact that the greedy algorithm applied to the design order is higher than that of $\frac{1}{2} P$ is exactly \cref{eq:half_char4}. 

It remains to show that the specified greedy allocation is incentive compatible, and corresponds to a serial dictatorship. Fix a profile of vertical types $\{v^i \}_{i\in I}$. Consider an arbitrary horizontal-type profile $\{h^i \}_{i\in I}$. Let $h^i[j]$ be the $j^{th}$ highest ranked item for agent $i$. By assumption, $h^i[j] = h^l[j]$ for all $i,l \in I$ and $j \in \{1, \dots, |N|\}$. Since $v \mapsto \nu_i(v|h)$ is increasing for all $i\in I$ and $h \in H$ by assumption, if $v_i > v_l$ then $\nu_i(v^i|h^i)h^i[j] \geq \nu_l(v^l|h^l)h^l[j]$ for all $j$. Thus for this fixed vertical type profile, the relative ranking of agents in the design order is the same, regardless of the horizontal type profile. Thus allocating the items according to the design order is equivalent to a serial dictatorship. 

The final piece is to show that this allocation rule is incentive compatible. This holds because for any agent $i$ and any horizontal-type profile, the rank of the item assigned to $i$ is increasing in $i$'s reported type. Thus $v_{\tau}> v_{\tau'}  \Rightarrow h_{\tau} \cdot \Vec{Q}_i(\tau) \geq h_{\tau'} \cdot \vec{Q}_i(\tau')$ holds, and the allocation is incentive compatible. 
\end{proof}

The situation is more delicate if horizontal types are unobserved. Assuming truthful reporting of horizontal types, the allocation quality of an individual is increasing in their vertical type report, and payments can be constructed such that truthful reporting of vertical types is optimal. The issue is that these payments will in general depend on the horizontal types: they are constructed such that downward IC constraints bind, but the value of downward deviations can depend on the horizontal type. Nonetheless, under stronger assumptions it is possible to show that the $R$-greedy allocation remains incentive compatible, albeit at the cost of replacing ex-post with interim incentive compatibility.

Say that horizontal types are \textit{uniform} if $h$ is a permutation of $\{1,\dots,N\}$

\begin{proposition}\label{prop:SD_unknownh}
    Assume 
    \begin{enumerate}[i.]
        \item $v,h$ are independent
        \item horizontal types are uniform and uniformly distributed
        \item and $v \mapsto v_i(v)$ is monotone for all $i$
    \end{enumerate}
    Then the $R$-greedy allocation is a $2$-approximation and is BIC (with suitable payments).  
\end{proposition}

\begin{proof}
    Assume that agents report horizontal types truthfully. Under the stated assumptions, for any profile of vertical type reports, payments are independent of the horizontal type reports. This is because in expectation, the change in an agents payoff from moving from rank $j$ in the priority order to rank $k$ in the priority order does not depend on their horizontal type. Then the only effect of mis-reporting the horizontal type is to induce a potentially worse assignment. The remainder of the argument is as in \Cref{prop:SD_knownh}.
\end{proof}

\newpage

\appendix
\section*{Appendix}

\section{Details on Birkhoff-von Neumann}\label{sec:birkhoff}

With indivisible items, an (ex-post) \textit{assignment} specifies to which agent, if any, each item goes. An ex-post assignment can be represented as an integer matrix $\rho \in \{0,1\}^{I\times N}$, where $\rho_{in}=1$ if $i$ is assigned item $n$, with the restriction that $\sum_{n\in N} \rho_{in} \leq 1$ for all $i\in I$ (the row constraint) and $\sum_{i\in I}\rho_{in} \leq 1$ for all $n\in N$ (the column constraint). Denote by $\mathcal{D}$ the set of assignments. Let $co(\mathcal{D})$ be the convex hull of the set of assignments. 

The design object is a map from type profiles to distributions over assignments. Denote this set by $\Delta(\mathcal{D})$. I restrict attention to problems in which only the marginals of such distributions matter.\footnote{Alternatively, in settings with fractional assignments, the design object is simply the quantity of each item assigned to each individual in each state.} That is, for each type profile $t$ both agents and the principal care only about the the marginal distribution over the item assigned to each agent.\footnote{In the single item case the distribution over assignments is uniquely identified by the marginal distribution, i.e. the probability that each agent gets the item. However with multiple items there may be multiple distributions over assignments which have the same marginals.} This rules out settings with complementarities. While the literature has largely focused on settings for which the marginal approach is appropriate, it is worth acknowledging that there are many interesting problems with complementarities.\footnote{Complementarities are know to introduce a number of difficulties into matching problems, such as non-existence of stable matches. See for example \cite{echenique2007solution} and \cite{che2019stable} for a discussion of matching with complements.} For example the principal's preference over which agent is assigned item $n$ may depend on which agent is assigned item $n'\neq n$. However, without the assumption that only the marginal assignment distribution matters, we could not adopt the interim approach to the problem, and would instead need to work directly with the allocation, an high-dimensional object. 

Under the assumption described above, the design object is simply an function from $T$ to $co(\mathcal{D})$. By the Birkhoff-von Neumann Theorem, any matrix $\rho \in \R_{+}^{I\times N}$ such that $\sum_{n\in N} \rho_{in} \leq 1$ for all $i\in I$ and $\sum_{i\in I}\rho_{in}\leq 1$ for all $n\in N$ is an element of $co(\mathcal{D})$.\footnote{For the case of general capacity constraints, this conclusion is implied by the generalization of the Birkhoff-von Neumann theorem in \cite{budish2013designing}.} In light of the above discussion then, define an \textit{allocation} as a map $q: T\times I\times N \rightarrow \R_+$ such that for all $t\in T$, $\sum_{i\in I}q(t,i,n) \leq 1$ for all $n\in N$ and $\sum_{n\in N} q(t,i,n) \leq 1$ for all $i\in I$. Here $q(t,i,n)$ is the probability that agent $i$ gets item $n$ when the type profile is $t$.

\section{Greedy algorithm and DA}\label{app:more_on_DA}

The equivalence between the greedy and DA assignments was established under the conditions that the principal's priority order be welfarist and item-ranking consistent. Here, I explore the types of priority orders for which these conditions hold, and the implications for the induced greedy/DA assignments. 

\begin{definition}
	An assignment $q$ is a \textbf{serial dictatorship} if there exists an order $M$ on $I$ such that $q$ is equivalent to the output of the following algorithm: order agents according to $M$, and have them pick their favorite item among those that have not already been chosen. 
\end{definition}

\begin{definition}
An assignment is \textbf{unresponsive} if it is invariant in the type profile. 
\end{definition}

\begin{lemma}
	If $q$ is a serial dictatorship then there exists a welfarist and item-ranking consistent priority order $\succ$ on $\mathcal{T}^*$ such that $q$ is the $\succ$-greedy allocation.
\end{lemma}
\begin{proof}
	Let agents be ordered according to $M$. We define $\succ$ lexicographically. Let $i_1$ be the first agent according to $M$. For each $\tau \in T_1$, order $N$ in $\tau$-decreasing order. Then arrange these lists, one for each $\tau \in T_1$, in any order. Do the same for $i_2$, the second agent according to $M$, and place this list after that for $i_1$. Proceeding in this way defines $\succ$. It is obviously welfarist. It is also items-ranking consistent by \Cref{lem:item_rankings}: let $>_n = M$ for all $n$.
\end{proof}

Welfarist and item-ranking consistent priority orders can generate assignments that combine elements of unresponsive and serial dictatorship assignments. The following example provides a concrete illustration.

\begin{example}
Let $I = \{i_1,i_2,i_3\}$ and $N = \{n_1,n_2,n_3\}$. Denote by $abc$ the type $\tau$ such that $n_1\tau n_2 \tau n_3$ (so type $123$ ranks $n_1$ first, then $n_2$, then $n_3$). Let $T_1 = \{123, 321 \}$, $T_2 = \{213, 321 \}$, and $T_3 = \{321 \}$. Consider the following priority order $\succ$, listed in decreasing order
\begin{align*}
(i_1, n_1, 123) & \\
\text{Segment } 1 \quad (i_1, n_2, 123) & \quad blocked\\ 
(i_1, n_3, 123) & \quad blocked\\ 
&\\
(i_2, n_2, 213) & \\
\text{Segment } 2 \quad (i_2, n_1, 213) & \quad blocked\\ 
(i_2, n_3, 213) & \quad blocked\\ 
&\\
(i_2, n_3, 321) & \\
\text{Segment } 3 \quad (i_2, n_2, 321) & \quad blocked\\ 
(i_2, n_1, 321) & \quad blocked\\ 
&\\
(i_3, n_3, 321) & \\
\text{Segment } 4 \quad (i_3, n_2, 321) & \\ 
(i_3, n_1, 321) & \quad blocked\\ 
&\\
(i_1, n_3, 321) & \quad blocked\\
\text{Segment } 5 \quad (i_1, n_2, 321) & \quad blocked\\ 
(i_1, n_1, 321) &
\end{align*}

 That $\succ$ is welfarist and item-ranking consistent is easily verified by inspection. The item priority orders $i_1 \succ_1 i_2 \succ_1 i_3$; $i_2 \succ_2 i_2 \succ_2 i_3$; and $i_2 \succ_3 i_3 \succ_3 i_1$ satisfy the conditions in \Cref{lem:item_rankings} for item-ranking consistency. Notice that the greedy algorithm applied to this order assigns $n_1$ to $i_1$ for any type profile. Moreover, $n_2$ and $n_3$ are allocated between $i_2$ and $i_3$ according to a serial dictatorship, with priority given to $i_2$.
\end{example}

\begin{remark}
There may be multiple principal priority orders that induce the same greedy allocation. In the above example, Segment 1 could be placed between Segments 4 and 5 without altering the greedy assignment, or violating welfarism and item-ranking consistency. Thus the restrictions the principal's priority orders implied by the joint assumption of welfarism and item-ranking consistency are less severe than the restrictions on the induced greedy assignments. 
 \end{remark}

\section{Alternative proof of \texorpdfstring{\Cref{thm:half_char}}{}}\label{app:half_char}
 
I first prove part $i$. Given this, I use the separating hyperplane theorem to show part $ii$. Part $iii.$ is immediate from part $ii.$

\vst
\noindent\textit{Proof of part $i.$} Fix an order $R$, and index $\mathcal{T}^*$ in $R$-decreasing order. Since $P$ satisfies condition (\ref{eq:BM}), 
\begin{equation}\label{eq:half_char1}
    \sum_{l = 1}^k P_{i_l}(\tau_l,n_l)\mu_{i_l}^{\bullet}(\tau_l) \leq \sum_{t\in T} \mu(t) c\left(I\left(t,\{(i_l,\tau_l,n_l)\}_{l=1}^k\right)\right)
\end{equation}
for all $k$. By construction of $q^R$, the projection of $q^R$ in the set $\{(i_l,\tau_l,n_l)\}_{l=1}^k$ covers $\{(i_l,\tau_l,n_l)\}_{l=1}^k$. Thus by \Cref{lem:halfbound}
\begin{equation}\label{eq:half_char2}
    \sum_{l = 1}^k Q^R_{i_l}(\tau_l,n_l)\mu_{i_l}^{\bullet}(\tau_l) \geq \frac{1}{2} \sum_{t\in T} \mu(t) c\left(I\left(t,\{(i_l,\tau_l,n_l)\}_{l=1}^k\right)\right).
\end{equation}
Part $i$ of \Cref{thm:half_char} follows by combining \cref{eq:half_char1} and \cref{eq:half_char2}. 

\vst
\noindent\textit{Proof of part $ii.$} By the separating hyperplane theorem, if $\frac{1}{2}P \not\in co(\mathcal{Q})$ then there exists $y \neq 0$ in $\R^{(\mathcal{T}^*\times N)}$ and $b \in \R$ such that 
\begin{equation}\label{eq:half_char3}
        \sum_{(i,\tau) \in \mathcal{T}^*}\sum_{n\in N} y(i,\tau,n)\mu_{i_l}^{\bullet}(\tau) \left(\frac{1}{2} P_i(\tau,n) - Q_i(\tau,n) \right) >0
\end{equation}
for any $Q \in \mathcal{Q}^M$. Order $\mathcal{T}^*$ so that $k \mapsto y(i_k,\tau_k,n_k)$ is decreasing, and let $R$ be the order that this corresponds to. Let $k^* = \max\{ k : y(i_k,\tau_k,n_k) > 0\}$. Note that $k^* < |\mathcal{T}^*\times N|$ since $0 \in \mathcal{Q}^M$. Let $\hat{Q}$ be the interim allocation induced by greedy allocation $q^R$ truncated at $k^*$. Then \cref{eq:half_char3} holds only if 
\begin{equation}\label{eq:half_char4}
\begin{split}
     0 &< \sum_{l=1}^{k^*} y(i_l,\tau_l,n_l)\mu_{i_l}^{\bullet}(\tau) \left(\frac{1}{2} P_{i_l}(\tau_l,n_l) - \hat{Q}_{i_l}(\tau_l,n_l) \right)\\
    & \leq \sum_{l = 1}^{k^*} \big(y(i_l,\tau_l,n_l) - y(i_{l+1},\tau_{l+1},n_{l+1}) \big)   \sum_{j=1}^l \mu_{i_j}^{\bullet}(\tau) \left(\frac{1}{2} P_{i_j}(\tau_j,n_j) - \hat{Q}_{i_j}(\tau_j,n_j)\right)    
\end{split}
\end{equation}
where we obtain the final inequality by rewriting the sum in the first line, and using the fact that  $y(i_{k^*+1},\tau_{k^*+1},n_{k^*+1}) \leq 0$ by assumption and $\frac{1}{2} P_{i_{k^*+1}}(\tau_{k^*+1},n_{k^*+1}) - \hat{Q}_{i_{k^*+1}}(\tau_{k^*+1},n_{k^*+1}) \geq 0$ by definition of $\hat{Q}$. By part $i$ we have that 
\begin{equation*}
    \left(\frac{1}{2} P_{i_j}(\tau_j,n_j) - \hat{Q}_{i_j}(\tau_j,n_j)\right) \leq 0
\end{equation*}
for all $l$, which implies that \cref{eq:half_char4} cannot hold.


\bibliography{references.bib}

\begin{thebibliography}{30}
\providecommand{\natexlab}[1]{#1}
\providecommand{\url}[1]{\texttt{#1}}
\expandafter\ifx\csname urlstyle\endcsname\relax
  \providecommand{\doi}[1]{doi: #1}\else
  \providecommand{\doi}{doi: \begingroup \urlstyle{rm}\Url}\fi

\bibitem[Abdulkadiro{\u{g}}lu and S{\"o}nmez(2003)]{abdulkadirouglu2003school}
A.~Abdulkadiro{\u{g}}lu and T.~S{\"o}nmez.
\newblock School choice: A mechanism design approach.
\newblock \emph{American economic review}, 93\penalty0 (3):\penalty0 729--747,
  2003.

\bibitem[Anshelevich and Sekar(2016)]{anshelevich2016truthful}
E.~Anshelevich and S.~Sekar.
\newblock Truthful mechanisms for matching and clustering in an ordinal world.
\newblock \emph{International Conference on Web and Internet Economics}, 2016.

\bibitem[Ba et~al.(2021)Ba, Bayer, Rim, Rivera, and Sidib{\'e}]{ba2021police}
B.~Ba, P.~Bayer, N.~Rim, R.~Rivera, and M.~Sidib{\'e}.
\newblock Police officer assignment and neighborhood crime.
\newblock Technical report, National Bureau of Economic Research, 2021.

\bibitem[Border(1991)]{border1991implementation}
K.~C. Border.
\newblock Implementation of reduced form auctions: A geometric approach.
\newblock \emph{Econometrica: Journal of the Econometric Society}, pages
  1175--1187, 1991.

\bibitem[Border(2007)]{border2007reduced}
K.~C. Border.
\newblock Reduced form auctions revisited.
\newblock \emph{Economic Theory}, 31\penalty0 (1):\penalty0 167--181, 2007.

\bibitem[Budish(2012)]{budish2012matching}
E.~Budish.
\newblock Matching" versus" mechanism design.
\newblock \emph{ACM SIGecom Exchanges}, 11\penalty0 (2):\penalty0 4--15, 2012.

\bibitem[Budish et~al.(2013)Budish, Che, Kojima, and
  Milgrom]{budish2013designing}
E.~Budish, Y.-K. Che, F.~Kojima, and P.~Milgrom.
\newblock Designing random allocation mechanisms: Theory and applications.
\newblock \emph{American economic review}, 103\penalty0 (2):\penalty0 585--623,
  2013.

\bibitem[Cai et~al.(2018)Cai, Daskalakis, and Weinberg]{cai2018constructive}
Y.~Cai, C.~Daskalakis, and S.~M. Weinberg.
\newblock A constructive approach to reduced-form auctions with applications to
  multi-item mechanism design.
\newblock \emph{arXiv preprint arXiv:1112.4572}, 2018.

\bibitem[Che et~al.(2013)Che, Kim, and Mierendorff]{che2013generalized}
Y.-K. Che, J.~Kim, and K.~Mierendorff.
\newblock Generalized reduced-form auctions: A network-flow approach.
\newblock \emph{Econometrica}, 81\penalty0 (6):\penalty0 2487--2520, 2013.

\bibitem[Che et~al.(2019)Che, Kim, and Kojima]{che2019stable}
Y.-K. Che, J.~Kim, and F.~Kojima.
\newblock Stable matching in large economies.
\newblock \emph{Econometrica}, 87\penalty0 (1):\penalty0 65--110, 2019.

\bibitem[Combe et~al.(2018)Combe, Tercieux, and Terrier]{combe2018design}
J.~Combe, O.~Tercieux, and C.~Terrier.
\newblock The design of teacher assignment: Theory and evidence.
\newblock \emph{Unpublished paper, University College London.[1310]}, 2018.

\bibitem[Dunstan and Welsh(1973)]{dunstan1973greedy}
F.~Dunstan and D.~Welsh.
\newblock A greedy algorithm for solving a certain class of linear programmes.
\newblock \emph{Mathematical Programming: Series A and B}, 5\penalty0
  (1):\penalty0 338--353, 1973.

\bibitem[Echenique and Yenmez(2007)]{echenique2007solution}
F.~Echenique and M.~B. Yenmez.
\newblock A solution to matching with preferences over colleagues.
\newblock \emph{Games and Economic Behavior}, 59\penalty0 (1):\penalty0 46--71,
  2007.

\bibitem[Gale and Shapley(1962)]{gale1962college}
D.~Gale and L.~S. Shapley.
\newblock College admissions and the stability of marriage.
\newblock \emph{The American Mathematical Monthly}, 69\penalty0 (1):\penalty0
  9--15, 1962.

\bibitem[Gershkov et~al.(2013)Gershkov, Goeree, Kushnir, Moldovanu, and
  Shi]{gershkov2013equivalence}
A.~Gershkov, J.~K. Goeree, A.~Kushnir, B.~Moldovanu, and X.~Shi.
\newblock On the equivalence of bayesian and dominant strategy implementation.
\newblock \emph{Econometrica}, 81\penalty0 (1):\penalty0 197--220, 2013.

\bibitem[Goeree and Kushnir(2022)]{goeree2022geometric}
J.~Goeree and A.~Kushnir.
\newblock A geometric approach to mechanism design.
\newblock \emph{Journal of Political Economy Microeconomics}, 2022.

\bibitem[Gopalan et~al.(2018)Gopalan, Nisan, and
  Roughgarden]{gopalan2018public}
P.~Gopalan, N.~Nisan, and T.~Roughgarden.
\newblock Public projects, boolean functions, and the borders of border's
  theorem.
\newblock \emph{ACM Transactions on Economics and Computation (TEAC)},
  6\penalty0 (3-4):\penalty0 1--21, 2018.

\bibitem[Hafalir et~al.(2013)Hafalir, Yenmez, and
  Yildirim]{hafalir2013effective}
I.~E. Hafalir, M.~B. Yenmez, and M.~A. Yildirim.
\newblock Effective affirmative action in school choice.
\newblock \emph{Theoretical Economics}, 8\penalty0 (2):\penalty0 325--363,
  2013.

\bibitem[Kleiner et~al.(2021)Kleiner, Moldovanu, and
  Strack]{kleiner2021extreme}
A.~Kleiner, B.~Moldovanu, and P.~Strack.
\newblock Extreme points and majorization: Economic applications.
\newblock \emph{Econometrica}, 89\penalty0 (4):\penalty0 1557--1593, 2021.

\bibitem[Lang et~al.(2022)Lang, Yang, et~al.]{lang2022reduced}
X.~Lang, Z.~Yang, et~al.
\newblock Reduced-form allocations for multiple indivisible objects under
  constraints: A revision.
\newblock Technical report, Working paper, 2022.

\bibitem[Manelli and Vincent(2010)]{manelli2010bayesian}
A.~M. Manelli and D.~R. Vincent.
\newblock Bayesian and dominant-strategy implementation in the independent
  private-values model.
\newblock \emph{Econometrica}, 78\penalty0 (6):\penalty0 1905--1938, 2010.

\bibitem[Maskin and Riley(1984)]{maskin1984optimal}
E.~Maskin and J.~Riley.
\newblock Optimal auctions with risk averse buyers.
\newblock \emph{Econometrica: Journal of the Econometric Society}, pages
  1473--1518, 1984.

\bibitem[Matthews(1984)]{matthews1984implementability}
S.~A. Matthews.
\newblock On the implementability of reduced form auctions.
\newblock \emph{Econometrica: Journal of the Econometric Society}, pages
  1519--1522, 1984.

\bibitem[Mierendorff(2011)]{mierendorff2011asymmetric}
K.~Mierendorff.
\newblock Asymmetric reduced form auctions.
\newblock \emph{Economics Letters}, 110\penalty0 (1):\penalty0 41--44, 2011.

\bibitem[Myerson(1981)]{myerson1981optimal}
R.~B. Myerson.
\newblock Optimal auction design.
\newblock \emph{Mathematics of operations research}, 6\penalty0 (1):\penalty0
  58--73, 1981.

\bibitem[Rochet(1987)]{rochet1987necessary}
J.-C. Rochet.
\newblock A necessary and sufficient condition for rationalizability in a
  quasi-linear context.
\newblock \emph{Journal of mathematical Economics}, 16\penalty0 (2):\penalty0
  191--200, 1987.

\bibitem[Roth(1985)]{roth1985college}
A.~E. Roth.
\newblock The college admissions problem is not equivalent to the marriage
  problem.
\newblock \emph{Journal of economic Theory}, 36\penalty0 (2):\penalty0
  277--288, 1985.

\bibitem[Shapley and Scarf(1974)]{shapley1974cores}
L.~Shapley and H.~Scarf.
\newblock On cores and indivisibility.
\newblock \emph{Journal of mathematical economics}, 1\penalty0 (1):\penalty0
  23--37, 1974.

\bibitem[Vohra(2011)]{vohra2011mechanism}
R.~V. Vohra.
\newblock \emph{Mechanism design: a linear programming approach}, volume~47.
\newblock Cambridge University Press, 2011.

\bibitem[Zheng(2022)]{zheng2022reduced}
C.~Zheng.
\newblock Reduced-form auctions of multiple objects.
\newblock Technical report, Working paper, 2022.

\end{thebibliography}

\end{document}